\documentclass[letterpaper, 10 pt, conference]{IEEEtran}  

\IEEEoverridecommandlockouts                              

\usepackage{tikz}
\usetikzlibrary{arrows,positioning,shapes.geometric}
\usetikzlibrary{snakes}
\usetikzlibrary{shapes,arrows}
\usepackage{svg}
\usepackage{amssymb}
\usepackage{amsthm}
\usepackage{amsmath}
\usepackage{hyperref}
\hypersetup{backref, colorlinks=true, linkcolor=blue, citecolor=blue, urlcolor = blue}
\usepackage{cleveref}
\usepackage{mathtools}
\usepackage{graphicx}
\usepackage{caption}
\newtheorem{theorem}{Theorem}
 
\newtheorem{lemma}{Lemma}
\newtheorem{proposition}{Proposition}

\newtheorem{definition}{Definition}

\title{\LARGE \bf
Optimal Causal Rate-Constrained Sampling of the Wiener Process
}

\author{Nian Guo and Victoria Kostina
\thanks{N. Guo and V. Kostina are with the Department
of Electrical Engineering, California Institute of Technology, Pasadena, CA, 91125 USA. E-mail: \{nguo,vkostina\}@caltech.edu. This work was supported in part by the National
Science Foundation (NSF) under grant CCF-1751356.}
}

\begin{document}

\maketitle
\thispagestyle{empty}
\pagestyle{empty}

\maketitle
\begin{abstract}
We consider the following communication scenario. An encoder causally observes the Wiener process and decides when and what to transmit about it. A decoder makes real-time estimation of the process using causally received codewords. We determine the causal encoding and decoding policies that jointly minimize the mean-square estimation error, under the long-term communication rate constraint of $R$ bits per second. We show that an optimal encoding policy can be implemented as a causal sampling policy followed by a causal compressing policy. We prove that the optimal encoding policy samples the Wiener process once the innovation passes either $\sqrt{\frac{1}{R}}$ or $-\sqrt{\frac{1}{R}}$, and compresses the sign of the innovation (SOI) using a 1-bit codeword. The SOI coding scheme achieves the operational distortion-rate function, which is equal to $D^{\mathrm{op}}(R)=\frac{1}{6R}$. Surprisingly, this is significantly better than the distortion-rate tradeoff achieved in the limit of infinite delay by the best non-causal code. This is because the SOI coding scheme leverages the free timing information supplied by the zero-delay channel between the encoder and the decoder. The key to unlock that gain is the event-triggered nature of the SOI sampling policy.  In contrast, the distortion-rate tradeoffs achieved with deterministic sampling policies are much worse: we prove that the causal informational distortion-rate function in that scenario is as high as $D_{\mathrm{DET}}(R) = \frac{5}{6R}$. It is achieved by the uniform sampling policy with the sampling interval $\frac{1}{R}$. In either case, the optimal strategy is to sample the process as fast as possible and to transmit 1-bit codewords to the decoder without delay. Finally, we show that the SOI coding scheme also minimizes the mean-square cost of a continuous-time control system driven by the Wiener process, and controlled via rate-constrained impulses.
\end{abstract}
\begin{IEEEkeywords}
Causal lossy source coding, sequential estimation, sampling.
\end{IEEEkeywords}

\section{Introduction}
\subsection{System Model}\label{IA}
Consider the system in Fig.~\ref{Fig1}. A source outputs a continuous-time standard Wiener process $\{W_t\}_{t= 0}^T$, within the time horizon $[0,T]$. An encoder observes the process and decides to disclose information about it at a sequence of non-decreasing codeword-generating time stamps
\begin{equation} \label{timestamp}
0\leq \tau_1\leq \tau_2\leq \dots\leq \tau_N\leq T.    
\end{equation}
These time stamps can be random and they can causally depend on the Wiener process. Consequently, the total number of time stamps $N$ can also be random. At time ${\tau}_i$, the encoder chooses to generate a binary codeword $U_i$, with a length ${\ell}_i\in \mathbb {Z}^+$, based on the past observed process $\{W_t\}_{t=0}^{\tau_i}$. Then, the codeword $U_i$ is passed through a noiseless digital channel to the decoder without delay. Upon receiving the codeword $U_i$ at time ${\tau}_i$, based on all the received codewords $U^i$ and the codeword-generating time stamps $\{\tau_1,\dots,\tau_i\}$, the decoder updates its running estimate of the Wiener process, yielding $\{\hat {W}_t\}_{t=\tau_i}^T$. The decoder updates its estimate $\{\hat {W}_t\}_{t={\tau}_{i+1}}^T$ once the next codeword $U_{i+1}$ is received at $\tau_{i+1}$.

\tikzstyle{int}=[draw, fill=white!20, minimum size=2.6em]
\tikzstyle{init} = [pin edge={to-,thin,black}]
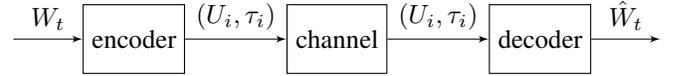
\begin{figure}[h!]
\centering
\begin{tikzpicture}[node distance=2.7cm,auto,>=latex']
   \node [int] (a) {encoder};
    \node (b) [left of=a,node distance=1.6cm, coordinate] {a};
    \node [int] (c) [right of=a] {channel};
    \node (d) [left of=c,node distance=2.1cm, coordinate] {c};
    \node [int] (e) [right of=c] {decoder};
    \node [coordinate] (end) [right of=e, node distance=1.6cm]{};
    \path[->] (b) edge node {$W_t$} (a);
    \path[->] (a) edge node {$(U_i,\tau_i)$} (c);
    \path[->] (c) edge node {$(U_i,\tau_i)$} (e);
    \draw[->] (e) edge node {$\hat W_t$} (end) ;
\end{tikzpicture}
\caption{System Model.} \label{Fig1}
\end{figure}

The communication between the encoder and the decoder is subject to a constraint on the long-term average transmission rate,
\begin{equation}\label{comm_cons}
    \frac{1}{T}\mathbb E\left[\sum_{i=1}^N\ell_i\right]\leq R ~\text{(bits per sec)}.
\end{equation}

The \emph{distortion} is measured by the long-term mean-square error (MSE) between $W_t$ and $\hat W_t$, $0\leq t\leq T$,
\begin{equation} \label{MSE}
      \frac{1}{T}\mathbb E \left[\int_{0}^T (W_t-\hat{W}_t)^2 dt\right]\leq d.
\end{equation}
We aim to find the jointly optimal encoding and decoding policies that achieve the best tradeoffs between the rate in \eqref{comm_cons} and the MSE in \eqref{MSE}.

\subsection{Literature Review}
Finding sampling policies at the encoder and estimation policies at the decoder to jointly minimize the end-to-end distortion under transmission constraints falls into the realm of optimal scheduling and sequential estimation problems. These problems are often encountered in network control systems, and has attracted significant research interest with the development of robotics, the Internet of things, and the smart grid.

\AA str\"{o}m and Bernhardsson \cite{1} compared uniform and symmetric threshold sampling policies\footnote{The symmetric threshold sampling corresponds to sampling the process if its current value exceeds or falls short of the previous sample by exactly a certain threshold.} (referred to as Riemann and Lebesgue sampling, respectively) in continuous-time first-order stochastic systems with a Wiener process disturbance, and showed that the Lebesgue sampling gives a lower distortion than the Riemann sampling under the same average sampling frequency. 
Imer and Ba\c{s}ar \cite{2} considered causal estimation of i.i.d. processes under MSE and the constraint on the total number of transmissions over a finite time horizon, and showed via dynamic programming, that the time-varying symmetric threshold sampling policy is optimal for i.i.d. Gaussian processes \cite{2}.
For causal estimation of multidimensional discrete-time Gauss-Markov processes, Cogill et al. \cite{cogill} aimed to find a sampling policy that minimizes a cost function combining the average weighted MSE and the average transmission cost over the infinite horizon. Cogill et al. \cite{cogill} proposed a threshold policy that transmits a sample once the absolute value of the squared error exceeds some constant, and proved that this suboptimal policy leads to a cost that is within a factor of 6 of the optimal achievable cost. 
Using dynamic programming and majorization theory, Lipsa and Martins \cite{3} proved that a time-varying symmetric threshold policy and a Kalman-like filter jointly minimize a discounted cost function consisting of MSE and a communication cost, for scalar discrete-time Gauss-Markov processes over a finite time horizon.
For partially observed discrete-time Gauss-Markov processes, Wu et al. \cite{Wu} fixed an event-triggered policy, where the encoder transmits only if the L-infinity norm of the Mahalanobis transformation of the measurement innovation exceeds a constant, and derived both the accurate and approximate (assuming Gaussian samples) minimum MSE (MMSE) estimator to combine with that sampling policy. Wu et al. \cite{Wu} also derived the relation between the transmission frequency and the threshold of the event-triggered policy. 
Rabi et al. \cite{4} formulated the problem of causal estimation of continuous-time scalar linear diffusion processes under the constraint on the total number of transmissions over a finite time horizon as an optimal stopping time problem. Rabi et al. \cite{4} solved the optimal stopping time problem iteratively to show that the optimal deterministic sampling policy for the Wiener process is a uniform sampling policy, and that the optimal event-triggered sampling policy is a time-varying symmetric threshold policy. Rabi et al. \cite{4} derived a dynamic program that can be used to numerically compute the optimal sampling policy for the Ornstein-Uhlenbeck process. 
Nar and Ba\c{s}ar \cite{5} extended the optimal stopping time problem in \cite{4} to the multidimensional Wiener process, and proved that a symmetric threshold policy remains optimal over both finite and infinite time horizons. In particular, Nar and Ba\c{s}ar \cite{5} showed that the optimal threshold over the infinite horizon is a constant depending on the average sampling frequency. 
For autogressive Markov processes driven by an i.i.d. process with unimodal and symmetric distribution, Charkravorty and Mahajan \cite{chk2} used ideas from renewal theory to prove that there is no loss of optimality if one focuses on sampling policies with time-homogeneous thresholds over the infinite time horizon. Charkravorty and Mahajan \cite{chk2} also proved that a symmetric threshold sampling policy together with an Kalman-like estimator achieves the distortion-transmission function, where the distortion-transmission function is defined as the minimum achievable expected average (or discounted) distortion subject to the expected average (or discounted) number of transmissions. 
For the same scenario as in \cite{chk2}, Molin and Hirche \cite{molin} proposed an iterative algorithm to show that the optimal event-triggered policy converges to a symmetric threshold policy.

In contrast to the scenarios in \cite{1}-\cite{molin}, where the communication channel is assumed to be perfect, \cite{6}-\cite{Ren} consider imperfect communication channels, such as a channel with an i.i.d. delay \cite{6}, a channel with i.i.d. Gamma noise \cite{7}, a fading channel \cite{Ren}, and a packet-drop channel \cite{chk3}.
Sun et al. \cite{6} proved that a symmetric threshold policy remains optimal even when the samples of the Wiener process experience an i.i.d. random transmission delay, but the threshold depends on the distribution of channel delay and is different from the one in \cite{5}. 
Using dynamic programming, Gao et al. in \cite{7} derived the optimal sampling, encoding and decoding policies for event-triggered sampling of an i.i.d. Laplacian source with subsequent transmission over a channel with a Gamma additive noise, under an average power constraint.
For discrete-time first-order autogressive Markov processes considered in \cite{chk2}-\cite{molin}, Ren et al. \cite{Ren} introduced a fading channel between the encoder and the decoder, where a successful transmission depends on both the channel gains and the transmission power, and found the optimal encoding and decoding policies that minimize an infinite horizon cost function combining the MSE and the power usage.
For first-order autogressive sources considered in \cite{chk2}\cite{molin}\cite{Ren}, Chakravorty and Mahajan \cite{chk3} further proved that the optimal estimation policy is a Kalman-like filter and the optimal sampling policy is symmetric threshold policy when the communication channel is a packet-drop channel with Markovian states, where the packet-drop probability depends on the channel state and the transmission power chosen by the encoder. 

Nayyar et al. \cite{Nayyar} considered a scenario where the encoder relies on the energy harvested from the environment to transmit messages to the estimator, with each transmission costing 1 unit of energy, and  proved that the optimal sampling strategy is a symmetric threshold policy, provided that the finite-state Markov source has symmetric and unimodal distribution and the distortion measure is either the Hamming distortion function or the power of the estimation error $|x-y|^p$. 
For the non-causal lossy source coding of a uniformly sampled Wiener process, Kipnis et al. \cite{8} derived the trade-offs between the sampling frequency, the communication bitrate and the estimation MSE, achievable in the limit of infinite delay.

Kofman and Braslavsky \cite{Kofman} designed a quantized event-triggered controller for noiseless partially observed continuous-time LTI systems with an unknown initial state to ensure asymptotic convergence of the system to the origin with zero average rate, seemingly violating the data-rate theorem. Similar to \cite{Kofman}, the fact that sampling time stamps of event-triggered policies carry information is also exploited in \cite{Pearson}-\cite{Kho}. Pearson et al. \cite{Pearson} considered encoding the deterministic and possibly nonuniformly sampled states of noiseless continuous-time LTI systems into symbols in a finite alphabet with a \emph{free} symbol representing the absence of transmission. For discrete-time linear systems with additive disturbances, Khina et al. \cite{23} considered a setting where at each discrete-time instant, the encoder chooses to either transmit 1 bit or transmit the free symbol, and designed an adaptive scalar quantizer with three bins using a Lloyd-Max algorithm with the quantization bin of the largest probability corresponding to the free symbol. Ling \cite{Ling} designed a periodic event-triggered quantization policy to stabilize continuous-time LTI systems subject to i.i.d. feedback dropouts, bounded network delay and bounded noise, which leads to a stabilizing rate that is lower than the one the data-rate theorem requires for time-triggered policies. Khojasteh et al. \cite{Kho} considered sampling noiseless continuous-time LTI systems when the state estimation error exceeds an exponentially decaying function, and found that the information transmission rate required for stabilizing systems can be any positive value for small enough delays and starts to increase when the delay exceeds a critical value.  Quantized event-triggered control has also been studied for continuous-time LTI systems with bounded disturbances \cite{Lehmann}, for partially-observed continuous-time LTI systems without noise \cite{Tawani} and with bounded noise \cite{Abdel}, for discrete-time noiseless linear systems \cite{Yoshi}, and for partially observed continuous-time LTI systems with time-varying network delay \cite{MIN}. Event-triggered control schemes to guarantee exponential stabilization were designed both for continuous-time LTI systems with bounded disturbances under a bounded rate constraint \cite{Talla2016} and for noiseless continuous-time LTI systems under time-varying rates constraints and channel blackouts \cite{Talla2019}.

\subsection{Contribution}
In this paper, we adopt an information-theoretic approach to continuous-time causal estimation, by considering the optimal tradeoff between the achievable MSE and the average number of bits communicated. This is different from the models studied in \cite{1}-\cite{Nayyar}, where communication cost is measured by the number of transmissions, and each infinite-precision transmission can carry an infinite amount of information. For communication over digital channels, a bitrate constraint, routinely considered in information theory, is more appropriate. Our setting is also different from \cite{8} in that we do not ignore delay: our distortion at time $t$ is measured with respect to the actual value of the process at time $t$; whereas \cite{8} permits an infinite delay, following a standard assumption in information theory. In contrast to the works \cite{Kofman}-\cite{Talla2019} that do not claim or consider the optimality of the proposed event-triggered policies, we show the optimality of the SOI coding scheme for our problem setting in Section~\ref{IA}.

We first show that an optimal encoding policy that achieves the operational distortion-rate function (ODRF) can be implemented as a causal sampling policy coupled with a compressing policy. Then, we prove that the optimal encoding policy is a symmetric threshold sampling policy with threshold $\pm \sqrt{\frac{1}{R}}$ and a 1-bit SOI compressor. The optimal decoding policy causally estimates the Wiener process by summing up the received innovations. This coding scheme, termed the SOI coding scheme, achieves the ODRF $D^{\mathrm{op}}(R)=\frac{1}{6R}$. 

In the SOI coding scheme, the encoder continuously tracks the process, generating a bit once the process passes the threshold. To reconstruct the process, both those bits and their time stamps are required at the decoder. In the scenario where, due to implementation constraint, the sampler is process-agnostic, or the decoder has no access to timing information, one has to adopt a deterministic sampling policy. We prove that a uniform sampling policy with the sampling interval $\frac{1}{R}$ achieves the informational distortion-rate function (IDRF), which is equal to $D_{\mathrm{DET}}(R)=\frac{5}{6R}$. To define the IDRF for the deterministic sampling policies, we change the rate constraint \eqref{comm_cons} to a directed mutual information rate constraint, which serves as an information-theoretic lower bound to \eqref{comm_cons}. This is a consequence of our real-time distortion constraint. Had we allowed delay, coding gains would have been possible by, for example, jointly compressing blocks of those bits. To confirm  that the IDRF is a meaningful gauge of what is achievable in the zero-delay causal compression, we implement the greedy Lloyd-Max compressor \cite{23} to compress the innovations $W_{\tau_i}-\hat W_{\tau_{i-1}}$, and verify that the performance of the resulting scheme is close to the IDRF. 

To study the tradeoffs between the sampling frequency and the rate per sample under a rate per second constraint $R$, we define operational and informational distortion-frequency-rate function (ODFRF and IDFRF). The ODFRF and the IDFRF are both minimized by the maximum sampling frequency $R$ and the minimum rate $1$ bit/sample, implying that sampling the process as fast as possible under the rate constraint and transmitting 1-bit codewords to the decoder without delay is optimal.

Surprisingly, the distortion achieved by the SOI coding scheme is smaller than the distortion achieved by the best non-causal codes. The reason is that in the SOI coding scheme, the encoder and the decoder know the random sampling time stamps perfectly, whereas in classical non-causal coding, this free timing information is not considered. 

We also show that the SOI coding scheme continues to be optimal when there is a fixed channel delay between the codeword-generating time and the codeword-delivery time. We show that if the decoder is allowed to wait for only the next codeword before decoding, the MSE can be further decreased.  

Finally, we prove that the SOI coding scheme is also optimal in a rate-constrained event-triggered control scenario with a continuous-time stochastic plant driven by the Wiener process and controlled via impulse control. The SOI code minimizes the mean-square cost between the state of the stochastic plant and the desirable state $0$.

 A part of this work will be presented at the 57th Annual Allerton Conference \cite{Nian}; the conference version does not contain Section~\ref{Control_sec} or any proofs.

\subsection{Paper organization}
In Section~\ref{SecII}, we define causal codes, distortion-rate and distortion-frequency-rate functions. In Section \ref{MR}, we state the main results of this paper, including the optimal causal sampling and compressing policies and the tradeoffs between the sampling frequency and the rate per sample. In Section~\ref{PMR}, we show the proof of the main results. In Section~\ref{V}, we discuss the distortion-rate tradeoffs when delays are allowed at both the encoder and the decoder, at the decoder only, and at the communication channel. In Section~\ref{Control_sec}, we show the optimal causal sampling and compressing policies in a rate-constrained event-triggered control system.

\subsection{Notations}
We denote by $\{W_t\}_{t=\tau_i}^{\tau_{i+1}}$ and $\{W_t\}_{\tau_i < t < \tau_{i+1}}$ the parts of the Wiener process within time intervals $[\tau_i,\tau_{i+1}]$, and $(\tau_i,\tau_{i+1})$, respectively. For $M\in\mathbb Z^+$, $[M]\triangleq\{1,\dots,M\}$. For a possibly infinite sequence $x=\{x_1,x_2,\dots\}$, we write $x^i=\{x_1,x_2,\dots,x_i\}$ to denote the vector of its first $i$ elements.

\section{Distortion-rate functions}\label{SecII}
In this section, we define the operational and the informational causal distortion-rate functions, and we show that an optimal encoder can be separated into a sampler followed by a compressor. 

\subsection{Encoding and decoding policies}\label{SecIIA}
The standard Wiener process is defined as follows. 
\begin{definition}\label{def1}
(standard Wiener process, e.g. \cite{13}) A standard Wiener process $\{W_t\}_{t\geq0}$ is a stochastic process characterized by the following three properties:

(i) \emph{time-homogeneity}: for all non-negative $s$ and $t$, $W_s$ and $W_{s+t}-W_t$ have the same distribution ($W_0=0$);

(ii) \emph{independent increments}: $W_{t_i}-W_{s_i}~(i\geq1)$ are independent whenever the intervals $(s_i,t_i]$ are disjoint;

(iii) $W_t$ follows the Gaussian distribution $\mathcal{N}(0, t)$.
\end{definition}
Throughout, we assume that both encoder and decoder know the initial state $W_0=0$ at $\tau_0=0$.

Next, we formally define the \emph{encoding} and \emph{decoding} policies\footnote{We refer to encoding and decoding \emph{policies} to emphasize their causal nature.}. 
Denote the set of continuous functions on the time interval $[0,t]$ by $\mathcal{C}_{[0,t]}$. Define the Wiener process stopped at a stopping time $\tau$ (e.g. \cite[Eq. 3.9]{25}) as:
\begin{equation}
W_t(\tau) =  
\begin{cases}
    W_t & \text{if}~ t\leq\tau\\
    W_{\tau} & \text{if}~ t>\tau. 
\end{cases}
\end{equation}
\begin{definition}\label{DEF2}(An $(R,d,T)$ causal code)
An $(R,d,T)$ causal code for the Wiener process $\{W_t\}_{t = 0}^T$ is a pair of encoding and decoding policies defined as follows.

The \emph{encoding} policy consists of

(i) the causal \emph{sampling} policy $\pi_T = \{\tau_1,\tau_2,\dots\}$ that decides the codeword-generating time stamps in \eqref{timestamp}
that are stopping times of the filtration $\sigma(\{W_t\}_{t=0}^T)$, and

(ii) the \emph{compressing} policy $f_T = \{\mathsf{f}_1,\mathsf{f}_2,\dots\}$\footnote{In some scenarios, we allow randomness in the mapping $\mathsf f_i$, replacing the deterministic mapping $\mathsf f_i$ in (5) by a transition probability kernel.},
\begin{equation} \label{enc}
    \mathsf f_i\colon \mathcal{C}_{[0,T]} \rightarrow \left[2^{\ell i}\right].
\end{equation}
 The codeword generated at time $\tau_i$ is $U_i = \mathsf f_i\left(\left\{ W_t(\tau_i)\right\}_{t=0}^T\right)$. The codewords' lengths must satisfy the long-term average rate constraint \eqref{comm_cons}. 

The \emph{decoding} policy causally maps the received codewords and the codeword-generating time stamps to a continuous-time process estimate $\{\hat W_t\}_{t=0}^T$ using
\begin{equation}\label{opt_dec}
    \hat W_t  \triangleq \mathbb E[W_t|U^i,\tau^i,
    t<\tau_{i+1}],~t\in[\tau_{i},\tau_{i+1}).
\end{equation}

Together, the encoding and the decoding policies must satisfy the long-term MSE constraint in \eqref{MSE}.
\end{definition}
In this work, we focus on the causal sampling policies satisfying the following natural assumptions:
\begin{itemize}
    \item[(i)] The sampling interval between any two consecutive stopping times, $\tau_{i+1}-\tau_i$, satisfies 
\begin{equation}\label{si}
    \mathbb E[\tau_{i+1}-\tau_i]<\infty,~i=0,1,\dots
\end{equation}
\item[(ii)] For all $i=0,1,\dots$, the conditional pdf $f_{\tau_{i+1}|\tau_i}$ exists.
\end{itemize}

The decoding policy in \eqref{opt_dec} forces the estimate $\hat W_t$ to be equal to the conditional expectation of $W_t$ given all the received information and the information that the next sample has not been transmitted yet. As we will show in the proof of Theorem~\ref{Thm4}, given the optimal sampling policy, the optimal decoding policy $\hat W_t$, $t\in[\tau_i,\tau_{i+1})$ can be simplified to the following equation,
\begin{equation}\label{sim_opt_dec}
    \hat W_t=\hat W_{\tau_i}\triangleq\mathbb E[W_{\tau_i}|U^i,\tau^i], ~t\in[\tau_i,\tau_{i+1}).
\end{equation}
Allowing more freedom in the design of a decoding policy cannot yield a lower MSE because \eqref{opt_dec} is the MMSE estimator of $W_t$ during $t\in[\tau_i,\tau_{i+1})$. This is a consequence of the zero-delay MSE constraint \eqref{MSE} at the decoder. As we explain in Section \ref{V.D.} below, had we allowed delay at the decoder, we could have improved performance by  e.g. using linear interpolation between recovered samples at the decoder. 

\subsection{Operational distortion-rate function}\label{SecIIB}
We now define the operational distortion-rate function.
\begin{definition}\label{DEF3}
(Operational distortion-rate function (ODRF))
The ODRF is the minimum distortion compatible with rate $R$ achievable by causal rate-$R$ codes in the limit of infinite time horizon:
\begin{equation}\label{1}
    D^{\mathrm{op}}(R) \triangleq \limsup_{T\rightarrow\infty} \inf\{d: \exists~ (R,d,T)~ \text{causal code}\}.
\end{equation}
\end{definition}
Equivalently, the ODRF is 
\begin{equation}
    D^{\mathrm{op}}(R)=\limsup_{T\rightarrow\infty}\inf_{\substack{\pi_T\in\Pi_T\\f_T\in F_T\colon\\ \eqref{comm_cons}}} \frac{1}{T}\mathbb E\left[\sum_{i=0}^N\int_{\tau_i}^{\tau_{i+1}}(W_t-\hat W_t)^2 dt\right],
\end{equation}
where $\tau_{N+1}\triangleq T$, and $\Pi_T$, $F_T$ denote the sets of all sampling and all compressing policies over the time horizon T respectively.

It turns out that the ODRF can be decomposed into the distortion due to sampling and the distortion due to quantization. 
\begin{proposition}\label{prop2}
The ODRF for the Wiener process can be written as
\begin{subequations}\label{2}
\begin{align} \label{6a}
    D^{\mathrm{op}}(R) =& \limsup_{T\rightarrow\infty}\inf_{\pi_T\in\Pi_T}\frac{1}{T}\Bigg\{\mathbb E  \Biggl[\sum_{i=0}^{N}\int_{\tau_i}^{\tau_{i+1}} (W_t - W_{\tau_i})^2 dt\Biggr]\\ \label{6b}
    &+\inf_{\substack{f_T\in F_T:\\~\eqref{comm_cons}}} \mathbb E \Biggl[ \sum_{i=1}^{N} (\tau_{i+1}-\tau_i)(W_{\tau_i}-\hat{W}_{\tau_i})^2 \Biggr]\Bigg\},
\end{align}
\end{subequations}
where $\hat W_{\tau_i}$ is given in \eqref{sim_opt_dec}.
Furthermore, if randomized compressing policies are allowed, there is no loss of optimality if at time $\tau_i$, a compressing policy only takes into account the innovation $W_{\tau_i}-\hat W_{\tau_{i-1}}$, past codewords $U^{i-1}$ and timing information $\tau^i$, rather than the whole process up to time $\tau_i$, as permitted by $\text{Definition}~\ref{DEF2}$. 
\end{proposition}
\begin{proof}
Appendix \ref{A}.
\end{proof}
In \eqref{6a}, $W_{\tau_i}$ is the MMSE estimator of $W_t$ at $t\in [\tau_i,\tau_{i+1})$, given the past lossless samples $\{W_{\tau_j}\}_{j=1}^{i}$ and the codeword-generating time stamps $\tau^i$. The expectation in \eqref{6a} is the sampling distortion due to causally estimating the Wiener process from its lossless samples $\{W_{\tau_j}\}_{j=1}^{i}$ taken under the sampling policy $\pi_T$. 

The expectation in \eqref{6b} is the mean-square quantization error of the samples, accumulated over sampling intervals of length $\tau_{i+1}-\tau_i,~i=1,\dots,N$. According to the compressing policy described in $\text{Proposition}~\ref{prop2}$, the minimization problem in \eqref{6b} is the operational zero-delay causal distortion-rate function of the discrete-time stochastic process formed by the samples. Furthermore, the encoding policy can be implemented as a sampler followed by a compressor. See Fig.~\ref{Fig2}.
\tikzstyle{int}=[draw, fill=white!20, minimum size=2.7em]
\tikzstyle{init} = [pin edge={to-,thin,black}]
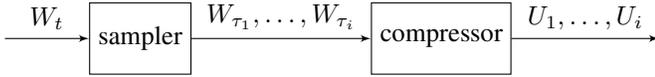
\begin{figure}[h!]
\centering
\begin{tikzpicture}[node distance=4cm,auto,>=latex']
    \node [int] (a) {sampler};
    \node (b) [left of=a,node distance=1.8cm, coordinate] {a};
    \node [int] (c) [right of=a] {compressor};
    \node [coordinate] (end) [right of=c, node distance=2.9cm]{};
    \path[->] (b) edge node {$W_t$} (a);
    \path[->] (a) edge node {$W_{\tau_1},\dots,W_{\tau_i}$} (c);
    \path[->] (c) edge node {$U_1,\dots,U_i$} (end) ;
\end{tikzpicture}
\caption{Decomposition of the encoder. }
\label{Fig2}
\end{figure}
The sampler takes measurements of the Wiener process under a sampling policy and outputs samples without delay to the compressor. Upon receiving a new sample, the compressor immediately generates a codeword under the compressing policy described in $\text{Proposition}~ \ref{prop2}$.

\subsection{Informational distortion-rate function}\label{SecIIC}
The \emph{directed information} $I(X^n\rightarrow Y^n)$ from a sequence $X^n$ to a sequence $Y^n$ is defined as \cite{14}
\begin{equation}
    I(X^n\rightarrow Y^n) = \sum_{i=1}^n I(X^i;Y_i|Y^{i-1}).
\end{equation}
The directed information captures the information due to the causal dependence of $Y^n$ on $X^n$.

A sampling policy $\pi_T = \{\tau_1,\tau_2,\dots\}$ is \emph{deterministic} if its sampling time stamps \eqref{timestamp} are deterministic. We denote the set of all deterministic sampling policies by $\Pi^{\mathrm{DET}}_T$. Under a deterministic sampling policy, the total number of samples $N$ within the time horizon $[0,T]$ is constant.
\begin{definition}\label{IDRF} (Informational distortion-rate function (IDRF)) The IDRF for the Wiener process under deterministic sampling policies can be written as
\begin{subequations}\label{idrf}
\begin{align} \nonumber 
    &D_{\mathrm{DET}}(R)\triangleq\\ \label{10a}
    & \limsup_{T\rightarrow\infty}\inf_{\pi_T\in \Pi^{\mathrm{DET}}_T}\frac{1}{T}\Bigg\{\mathbb E  \biggl[\sum_{i=0}^{N}\int_{\tau_i}^{\tau_{i+1}} (W_t - W_{\tau_i})^2 dt\biggr]+\\  \label{9b}
    &\inf_{\substack{\bigotimes_{i=1}^{N}P_{\hat W_{\tau_i}|W^{\tau_i},\hat{W}^{\tau_{i-1}}}:\\\frac{ I(W^{\tau_{N}}\rightarrow\hat W^{\tau_{N}})}{T}\leq R}}  \mathbb E \biggl[\sum_{i=1}^{N} (\tau_{i+1}-\tau_i)(W_{\tau_i}-\hat{W}_{\tau_i})^2 \biggr]\Bigg\},
\end{align}
\end{subequations}
\end{definition}
The minimization problem \eqref{9b} in $D_{\mathrm{DET}}(R)$ is the causal IDRF for the discrete-time stochastic process formed by the samples. Note that \eqref{9b} is minimized over the directed information rate, which gives an information-theoretic lower bound to the rate considered in \eqref{comm_cons}. Thus, the following relation holds according to \cite[Eq. (43)]{28}.
\begin{equation}\label{oplarger}
    D^{\mathrm{op}}_{\mathrm{DET}}(R)\geq D_{\mathrm{DET}}(R),
\end{equation}
where $D^{\mathrm{op}}_{\mathrm{DET}}(R)$ is the ODRF for deterministic sampling policies defined by \eqref{2} with the minimization constraint in \eqref{6a} replaced by $\pi_T\in \Pi^{\mathrm{DET}}_T$.

\subsection{Operational and informational distortion-frequency-rate function}\label{SecIID}
According to $\text{Proposition}~ \ref{prop2}$, an optimal encoder can be implemented as a sampler followed by a compressor. To gain insight into the tradeoffs between the sampling frequency $f$ at the sampler and the rate per sample $R_s$ at the compressor, we define an $(f,R_s,d,T)$ causal code.
\begin{definition}
(An $(f,R_s,d,T)$ causal code) An $(f,R_s,d,T)$ causal code for the Wiener process $\{W_t\}_{t=0}^T$ is a triplet of causal sampling, compressing and decoding policies:

(i) the causal sampling policy\footnote{The causal sampling policy is defined in Definition~\ref{DEF2}(i)} $\pi_T = \{\tau_1,\tau_2,\dots\}$ satisfies the average sampling frequency constraint 
\begin{equation}\label{f}
  \frac{1}{T} \mathbb E{[N]}=f; 
\end{equation}

(ii) the compressing policy $f_T=\{\mathsf f_1,\mathsf f_2,\dots\}$\footnote{Here we slightly abuse the notation: we have used $f_T$ in Definition~\ref{DEF2}(ii), and have shown in Proposition~\ref{prop2} that the compressing policy $f_T$ can be simplified to \eqref{f_new}.} is
\begin{equation}\label{f_new}
    \mathsf f_i: \mathbb R \times \mathbb R^{i-1} \times \mathbb R^i \rightarrow \left[2^{\ell_i}\right].
\end{equation}
The codeword generated at time $\tau_i$ is $U_i=\mathsf f_i\left(W_{\tau_i},U^{i-1},\tau^i\right)$. The codewords' lengths must satisfy 
\begin{equation}\label{18a}
 \frac{1}{\mathbb E[N]}\mathbb E\left[\sum_{i=1}^N\ell_i\right]\leq R_s~ (\text{bits per sample});   
\end{equation}

(iii) the decoding policy causally maps the received codewords and the codeword-generating time stamps to a continuous-time process estimate $\{\hat W_t\}_{t=0}^T$ using \eqref{opt_dec}.

Together, the causal sampling, compressing and decoding policies must satisfy the long-term MSE constraint in \eqref{MSE}.
\end{definition}

We define the operational distortion-frequency-rate function.
\begin{definition}
(Operational distortion-frequency-rate function(ODFRF)) The ODFRF is the minimum distortion achievable by causal frequency-$f$ and rate-$R_s$ codes in the limit of infinite time horizon:
\begin{equation}
    D^{\mathrm{op}}(f,R_s) \triangleq \limsup_{T\rightarrow\infty} \inf\{d:\exists~ (f,R_s,d,T)~ \text{causal~code}\}.
\end{equation}
\end{definition}

Using the method used to decompose $D^{\mathrm{op}}(R)$ in Proposition~\ref{prop2}, we can write $D^{\mathrm{op}}(f,R_s)$ as
\begin{subequations}\label{odfrf}
\begin{align}\nonumber
  &D^{\mathrm{op}}(f,R_s) = \\ \label{odfrf_a}
  &\limsup_{T\rightarrow\infty} 
\inf_{\substack{\pi_T\in\Pi_T\colon\\ \eqref{f}}}\frac{1}{T}\Bigg\{\mathbb E  \Biggl[\sum_{i=0}^{N}\int_{\tau_i}^{\tau_{i+1}} (W_t - W_{\tau_i})^2 dt\Biggr]\\ \label{odfrf_b}
    &+\inf_{\substack{f_T\in F_T:\\\eqref{18a}}} \mathbb E \Biggl[ \sum_{i=1}^{N} (\tau_{i+1}-\tau_i)(W_{\tau_i}-\hat{W}_{\tau_i})^2 \Biggr]\Bigg\},
\end{align}
\end{subequations}
where the expectation in \eqref{odfrf_a} is the sampling distortion, and the expectation in \eqref{odfrf_b} is the mean-square quantization error of the samples weighted by the lengths of sampling intervals $\tau_{i+1}-\tau_i$, $i=1,\dots,N$.

We define the informational distortion-frequency-rate function for deterministic sampling policies. The informational equivalent of $D^{\mathrm{op}}(f,R_s)$ replaces \eqref{18a} by the constraint on the directed information, that is, for deterministic sampling policies, 
\begin{equation}\label{18b}
    \frac{1}{N}I(W^{\tau_N}\rightarrow\hat W^{\tau_N})\leq R_s.
\end{equation}
\begin{definition}
(Informational distortion-frequency-rate function (IDFRF)) The IDFRF for the Wiener process under deterministic sampling policies can be written as
\begin{subequations}\label{idfrf}
\begin{align}\nonumber
    &D_{\mathrm{DET}}(f,R_s) \triangleq\\  \label{idfrf_a} 
    &\limsup_{T\rightarrow\infty}
    \inf_{\substack{\pi_T\in \Pi^{\mathrm{DET}}_T\colon\\ \eqref{f}}}\frac{1}{T}\Bigg\{\mathbb E  \biggl[\sum_{i=0}^{N}\int_{\tau_i}^{\tau_{i+1}} (W_t - W_{\tau_i})^2 dt\biggr]\\  \label{idfrf_b}
    &+\inf_{\substack{\bigotimes_{i=1}^{N}P_{\hat W_{\tau_i}|W^{\tau_i},\hat{W}^{\tau_{i-1}}}:\\ \eqref{18b}}}  \mathbb E \biggl[\sum_{i=1}^{N} (\tau_{i+1}-\tau_i)(W_{\tau_i}-\hat{W}_{\tau_i})^2 \biggr]\Bigg\}
\end{align}
\end{subequations}
\end{definition}
Similar to $D_{\mathrm{DET}}(R)$ in Definition~\ref{DEF3}, \eqref{idfrf_b} is the IDRF for the Gauss-Markov process formed by the samples, but it is worth noticing that the rate considered in \eqref{idfrf_b} is the rate per sample $R_s$ rather than the rate per second $R$ considered in \eqref{9b}.

\section{Main Results}\label{MR}
The first theorem of this section shows the optimal causal sampling and compressing policies that achieve $D^{\mathrm{op}}(R)$.
\begin{theorem}\label{Thm4}
In causal coding of the Wiener process, the optimal causal sampling policy is the following symmetric threshold sampling policy:
 \begin{equation}\label{opt_sym}
     \tau_{i+1} = \inf\left\{t\geq \tau_{i}: |W_{t}-W_{\tau_i}|\geq \sqrt{\frac{1}{R}}\right\},\; i = 0,1,2,\dots
 \end{equation}
 
The optimal compressing policy is a 1-bit sign-of-innovation (SOI) compressor:
 \begin{equation} \label{SOI}
U_i=
    \begin{cases} 
      1 & \text{if}\quad W_{\tau_{i+1}}-W_{\tau_{i}}\geq 0  \\
      0 & \text{if}\quad W_{\tau_{i+1}}-W_{\tau_{i}}< 0.
   \end{cases}
\end{equation}

 The SOI coding scheme achieves the ODRF:
\begin{equation}\label{58}
 D^{\mathrm{op}}(R)=\frac{1}{6R}.
 \end{equation}
\end{theorem}
\begin{proof}
Section~\ref{IV.A}. 
\end{proof}
Together with the optimal encoding policy in Theorem~\ref{Thm4}, the optimal decoding policy \eqref{opt_dec} accumulates the received noiseless innovations to estimate the current value of the process.

The next theorem shows the optimal deterministic sampling policy that achieves $D_{\mathrm{DET}}(R)$.
\begin{theorem}\label{Thm5}
In causal coding of the Wiener process, the uniform sampling with the sampling interval equal to
\begin{equation}
     \tau_{i+1}-\tau_i = \frac{1}{R},~i=0,1,2,\dots,
 \end{equation}
 achieves
\begin{equation}\label{Det}
   D_{\mathrm{DET}}(R) =\frac{5}{6R}.
\end{equation}
\end{theorem}
\begin{proof}
Section~\ref{IV.D}.
\end{proof}

\begin{theorem}\label{thm6}
In causal coding of the Wiener process, the ODRF satisfies
\begin{subequations}\label{25aa}
\begin{align}\label{25a_a}
   D^{\mathrm{op}}(R) =& \min_{\substack{f>0,R_s\geq 1\colon\\fR_s\leq R}}D^{\mathrm{op}}(f,R_s),\\ \label{25a_b}
   =&D^{\mathrm{op}}(R,1),
\end{align}
\end{subequations}
and the IDRF under deterministic sampling policies satisfies
\begin{subequations}\label{25bb}
\begin{align}  \label{25a_f}
   D_{\mathrm{DET}}(R) =& \min_{\substack{f>0,R_s\geq 1\colon\\ fR_s\leq R}} D_{\mathrm{DET}}(f,R_s)\\ \label{25a_fb}
   =&D_{\mathrm{DET}}(R,1).
\end{align}
\end{subequations}

\end{theorem}
\begin{proof}
See Section~\ref{IV-B} for the proof of \eqref{25aa}. See Section~\ref{IV.B} for the proof of \eqref{25bb}.
\end{proof}
Using Theorem~\ref{thm6}, we can formulate the \emph{working principle} of an optimal encoding policy as follows. A sampler takes measurements of the Wiener process as fast as possible subject to a rate constraint, and the most recent sample is used to generate a 1-bit codeword, which is transmitted to the decoder without delay. In the setting of $\text{Theorem}~\ref{Thm4}$, the 1-bit SOI compressor associated with the symmetric threshold sampling policy uses the most recent sample to calculate the innovation and to produce a 1-bit codeword. In the setting of Theorem~\ref{Thm5}, although evaluating $D_{\mathrm{DET}}(R)$ does not give us an operational compressing policy, we know that the stochastic kernel that achieves the causal IDRF for discrete-time Gauss-Markov processes formed by the samples under uniform sampling policies has the form $\bigotimes_{i=1}^\infty P_{\hat{W}_{\tau_i}|W_{\tau_i}-\hat{W}_{\tau_{i-1}},\hat{W}_{\tau_{i-1}}}$ \cite[Eq. (5.12)]{27}, suggesting that at the encoder, it is sufficient to compress the quantization innovation $ W_{\tau_i}-\hat W_{\tau_{i-1}}$ only. The decoder computes the estimate $\hat W_{\tau_i}$ as $\hat W_{\tau_i} = \hat W_{\tau_{i-1}}+\mathsf{q}_i(W_{\tau_i}-\hat W_{\tau_{i-1}})$, where $\mathsf{q}_i = \mathsf{g}_i\circ \mathsf{f}_i$, $\mathsf{f}_i\left(W_{\tau_i}-\hat{W}_{{\tau}_{i-1}}\right)\in\left[2^{\ell_i}\right]$ is the $i$-th binary codeword, and $\mathsf{g}_i(c)\in\mathbb R$ is the quantization representation point corresponding to $c\in \left[2^{\ell_i}\right]$. In practice, one can use the \emph{greedy Lloyd-Max compressor} \cite{23} that runs the Lloyd-Max algorithm for the quantization innovation in each step based on the prior probability of the quantization innovation. Specifically, the prior for $(i+1)$-th step is the pdf of the quantization innovation $W_{\tau_{i+1}}-\hat W_{\tau_{i}}$, which can be computed as the convolution of the pdfs of the quantization error $W_{\tau_{i}}-\hat W_{\tau_i}$ and the process increment $W_{\tau_{i+1}}-W_{\tau_i}$. The globally optimal scheme has a negligible gain over the greedy Lloyd-Max algorithm even in the finite time horizon \cite{23}.

\begin{figure}[h!]
\centering
 \includegraphics[width=0.45\textwidth]{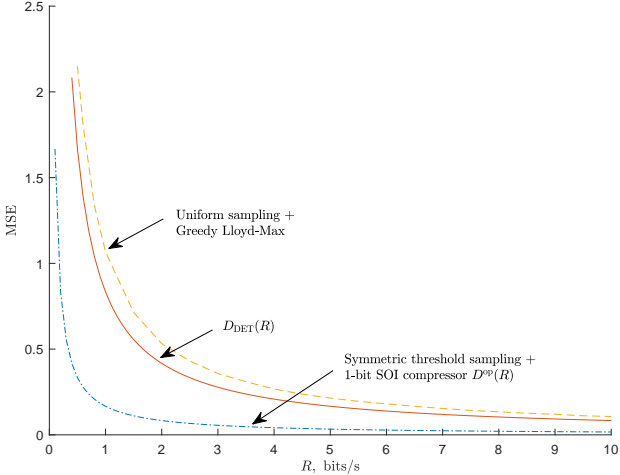}
  \caption{MSE versus rate}
  \label{Fig3}
\end{figure}

Fig.~\ref{Fig3} displays distortion-rate tradeoffs obtained in $\text{Theorems}~\ref{Thm4}~\text{and} ~\ref{Thm5}$, as well as a numerical simulation of the uniform sampling in $\text{Theorem}~ \ref{Thm5}$ with the greedy Lloyd-Max quantization of innovations. The symmetric threshold sampling policy followed by a 1-bit SOI compressor leads to a much lower MSE than uniform sampling. Indeed, according to $\text{Theorems}~\ref{Thm4}~\text{and} ~\ref{Thm5}$, $\frac{D_{\mathrm{DET}}(R)}{D^{\mathrm{op}}(R)}=5$, and $D^{\mathrm{op}}_{\mathrm{DET}}(R)$ for the uniform sampling is even higher than $D_{\mathrm{DET}}(R)$ by \eqref{oplarger}. Note that the greedy Lloyd-Max curve is rather close to the $D_{\mathrm{DET}}(R)$ curve, indicating that the IDRF is a meaningful gauge of what is attainable in zero-delay continuous-time causal compression.

The optimal sampling policies of $\text{Theorems}~\ref{Thm4}~\text{and}~ \ref{Thm5}$, i.e. the symmetric threshold and the uniform sampling policies, are the same as the corresponding optimal sampling policies that achieve the minimum sampling distortion \cite[Sec. 3.1]{4} \cite{5} subject to an average sampling frequency constraint \eqref{f} with $f=R$.
The value of $D^{\mathrm{op}}(R)$ \eqref{58} achieved by the symmetric threshold sampling policy is the same as the sampling distortion, since the 1-bit SOI compressor is able to compress each innovation noiselessly due to the size-2 alphabet of the innovations, resulting in zero quantization distortion \eqref{6b}. In contrast, for deterministic sampling policies, quantization distortion is unavoidable, since the samples are Gaussian. If we only consider the constraint on the sampling frequency, the optimal deterministic sampling policy for the Wiener process is uniform sampling \cite[Sec. 3.1]{4}. Nevertheless, the result in Theorem~\ref{Thm5} implies that uniform sampling is still optimal in the IDRF sense, whether or not the quantization distortion is taken into account.

\section{Proofs of the Main Results}\label{PMR}
\subsection{Proof of Theorem~\ref{Thm4}}\label{IV.A}
To prove that the SOI coding scheme in Theorem~\ref{Thm4} achieves the ODRF, we first derive a lower bound to the ODRF, and then we show that this lower bound is achieved by the SOI coding scheme. The MSE achievable by causal rate-$R$ codes is lower bounded in the following way,
\begin{subequations}
\allowdisplaybreaks
\begin{align} \label{thm1_dec_b} 
    &\inf_{\substack{ \pi_T\in\Pi_T,\\ f_T\in F_T\colon\\ \eqref{comm_cons}}}\frac{1}{T}\mathbb{E}\biggl[\sum_{i=0}^N\int_{\tau_i}^{\tau_{i+1}} (W_t- \mathbb E[W_t|U^i,\tau^i,t<\tau_{i+1}])^2dt\biggr]\\ \nonumber
    \geq &\inf_{\substack{ \pi_T\in\Pi_T\colon\\ \frac{\mathbb E[N]}{T}\leq R}}\frac{1}{T}\mathbb{E}\biggl[\sum_{i=0}^N\int_{\tau_i}^{\tau_{i+1}} (W_t- \mathbb E[W_t|\{W_s\}_{s=0}^{\tau_i},\tau^i,\\ \label{thm1_dec_b_f}  
    &t<\tau_{i+1}])^2dt\biggr]\\ \label{thm1_dec_c} 
    =&\inf_{\substack{ \pi_T\in\Pi_T\colon\\ \frac{1}{T}\mathbb E[N]\leq R}}\frac{1}{T}\mathbb{E}\biggl[\sum_{i=0}^N\int_{\tau_i}^{\tau_{i+1}} (W_t- \mathbb E[W_t|W_{\tau_i},\tau_i])^2dt\biggr]\\ \label{thm1_dec_d} 
    =&\inf_{\substack{ \pi_T\in\Pi_T\colon\\ \frac{1}{T}\mathbb E[N]\leq R}}\frac{1}{T}\mathbb{E}\biggl[\sum_{i=0}^N\int_{\tau_i}^{\tau_{i+1}} (W_t- W_{\tau_i})^2dt\biggr],
\end{align}
\end{subequations}
where \eqref{thm1_dec_b_f} holds since $\sigma(U^i)\subset\sigma(\{W_t\}_{t=0}^{\tau_i})$ and 
\begin{align}\label{NL}
    \mathbb E[N]\leq \mathbb E\left[\sum_{i=1}^N\ell_i\right] ;
\end{align}
\eqref{thm1_dec_c} holds due to \cite[Cor. 1.1]{Nian2} where the Wiener process satisfies the regularity conditions in \cite{Nian2}; \eqref{thm1_dec_d} is due to the strong Markov property of the Wiener process. 

It remains to show that the lower bound \eqref{thm1_dec_d} is achieved by the SOI coding scheme. First, we notice that the optimization problem in \eqref{thm1_dec_d} corresponds to determining the optimal sampling policy that minimizes the MSE subject to an average sampling frequency constraint, where $N$ can be considered as the total number of samples taken within $[0,T]$. According to \cite[Eq. (20)]{5}, the optimal sampling policy that achieves the $\limsup_{T\rightarrow\infty}$ of \eqref{thm1_dec_d} is given by \eqref{opt_sym}.
Second, we know that each innovation 
\begin{equation}
    \Delta W_{i}\triangleq W_{\tau_{i+1}}-W_{\tau_i},~ i=0,1,2,\cdots,
\end{equation}
is equiprobabily distributed on a size-2 alphabet $\left\{\pm \sqrt{\frac{1}{R}}\right\}$. Thus, $\Delta W_{i}$ can be noiselessly encoded using 1-bit codewords $U^i$, while satisfying the inequality in \eqref{NL} with equality. Therefore, the $\limsup_{T\rightarrow\infty}$ of \eqref{thm1_dec_d} is achieved by the SOI coding scheme.

From the equality in \eqref{thm1_dec_c} and the fact that the SOI coding scheme attains \eqref{thm1_dec_b}, we conclude that the optimal decoding policy \eqref{opt_dec} can indeed be simplified to \eqref{sim_opt_dec} given the optimal encoding policy in Theorem~\ref{Thm4}.

\subsection{Proof of Theorem~\ref{thm6} \eqref{25aa}}\label{IV-B}
$D^{\mathrm{op}}(f,R_s)$ is lower-bounded by the sampling distortion \eqref{odfrf_a}. This lower bound is achieved by a symmetric threshold sampling policy with thresholds $\pm \sqrt{\frac{1}{f}}$, and a 1-bit SOI compressor, where the symmetric threshold sampler achieves the minimum of \eqref{odfrf_a}, which is equal to $\frac{1}{6f}$ \cite[Eq. (20)]{5}. Since the 1-bit SOI compressor results in zero quantization distortion \eqref{odfrf_b},
\begin{equation}\label{58_f}
    D^{\mathrm{op}}(f,R_s) = \frac{1}{6f},
\end{equation}
for any $R_s\geq 1$. Plugging \eqref{58_f} into the minimization problem in \eqref{25a_a}, we obtain
\begin{subequations}\label{25a_a_min}
\begin{align}
    \min_{\substack{f>0,R_s\geq 1\colon\\fR_s\leq R}} D^{\mathrm{op}}(f,R_s) &= D^{\mathrm{op}}(R,1),\\
    D^{\mathrm{op}}(R,1)&= \frac{1}{6R}.
\end{align}
\end{subequations}

Comparing \eqref{25a_a_min} to \eqref{58}, we conclude that \eqref{25aa} holds.

\subsection{Proof of Theorem~\ref{thm6} \eqref{25bb} }\label{IV.B}

Since the samples taken under a deterministic sampling policy form a Gauss-Markov process, we first compute $D_{\mathrm{DET}}(f,R_s)$ building on existing results on the causal IDRF \eqref{idfrf_b} of discrete-time Gauss-Markov processes.

\begin{lemma}\label{6}
The IDFRF under deterministic sampling policies can be written as
\begin{subequations} \label{24}
\begin{align}\label{37aa}
&D_{\mathrm{DET}}(f,R_s) = \limsup_{N\rightarrow\infty} D_N(f,R_s),\\ \label{24bb}
&D_N(f,R_s)=\inf_{\substack{T^N\geq 0\colon\\ \eqref{f_equi}}}\frac{f}{N}\left(\sum_{i=0}^{N}\frac{T_i^2}{2}+\min_{\substack{D^N\geq 0\colon\\\eqref{dfrs_cons}}} \sum_{i=1}^{N}T_iD_i\right),
\end{align}
\end{subequations}
where the minimization constraints in \eqref{24} are
\begin{equation}\label{f_equi}
    \frac{1}{N}\sum_{i=0}^N T_i=\frac{1}{f},
\end{equation}
and
\begin{subequations}\label{dfrs_cons}
\begin{align}\label{dfrs_cons_a}
&z\left(D^N\right)\triangleq\frac{1}{N}\left(\sum_{i=1}^{N-1}\log\left(1+\frac{T_i}{D_i}\right)+\log\left(\frac{T_0}{D_N}\right)\right)\leq 2R_s,\\ \label{min2}
&D_{i-1}+T_{i-1}\geq D_i, i=1,\dots,N.
\end{align}
\end{subequations}
\end{lemma}
\begin{proof}
Appendix \ref{D}.
\end{proof}

The optimization variable $T^N$ in \eqref{24} is the vector of sampling intervals $T^N = \{T_0, T_1, \ldots, T_N\}$, where 
\begin{equation}
\begin{aligned}
    &T_i = \tau_{i+1} - \tau_i, ~i=0,\dots,N-1,\\
    &T_N = T-\tau_N,
\end{aligned}
\end{equation}
that determine a deterministic sampling policy. The optimization variable $D^N$ in \eqref{24} is the vector of sample distortions $D^N = \{D_1,\dots,D_N\}$.

Note that $D_{\mathrm{DET}}(R)$ in \eqref{idrf} is related to $D_{\mathrm{DET}}(f,R_s)$ in \eqref{24} as follows,
\begin{equation}\label{IDRF_IDFRF}
    D_{\mathrm{DET}}(R) = \limsup_{N\rightarrow\infty} \inf_{\substack{f>0,R_s\geq 1\colon\\ fR_s\leq R}} D_{N}(f,R_s).
\end{equation}
We observe that \eqref{25a_f} does not directly follow \eqref{IDRF_IDFRF}, since the right-hand side of \eqref{25a_f} switches the order of $\limsup$ and $\inf$ in \eqref{IDRF_IDFRF}.

We will use Lemmas~\ref{Lemm2}~to~\ref{Lemm5} that follow to prove \eqref{25a_f} in Theorem~\ref{thm6}.
\begin{lemma}\label{Lemm2}
$D_N(f,R_s)$ is lower-bounded as
\begin{subequations}\label{476}
\begin{align}
&D_N(f,R_s) \geq \underline{D}_N(f,R_s),\\ \nonumber
&\triangleq\inf_{\substack{T_0\geq0,T_N\geq 0\\ T_0+T_N\leq \frac{N}{f}}}
    \frac{f}{2} \Biggl(\frac{T_0^2+T_N^2+2\log e\lambda^*(f,R_s,N)}{N}\\\label{38bb}
    &+\frac{N-1}{N}T^*(f,N)\sqrt{T^*(f,N)^2+4\log e\lambda^*(f,R_s,N)}\Biggr),
\end{align}
\end{subequations}
where $T^*(f,N)$ is given by,
\begin{equation}\label{T_seq}
      T^*(f,N) \triangleq \frac{N}{f(N-1)}-\frac{T_0+T_N}{N-1}, i=1,\dots,N-1,
\end{equation}
and $\lambda^*(f,R_s,N)\geq 0$ is the unique solution to
\begin{equation} \label{lambda_exist}
    z\left(D^{N*}\right) = 2R_s,
\end{equation}
with $D^N$ in \eqref{dfrs_cons_a} replaced by
\begin{subequations}\label{28ab}
\begin{align}
    &D_i^* = \frac{-T_i+\sqrt{T_i^2+4\log e\lambda^*(f,R_s,N)}}{2},~i=1,\dots,N-1,\\
    &D_N^* = \frac{\lambda^*(f,R_s,N)\log e}{T_N},
\end{align}
\end{subequations}
and $T_i$, $i=1,\dots,N-1$ in \eqref{dfrs_cons_a} replaced by $T^*(f,N)$ in \eqref{T_seq}.
\end{lemma}
\begin{proof}
Appendix~\ref{proofLemm2}.
\end{proof}

\begin{lemma}\label{Lemm3}
$D_N(f,R_s)$ is upper-bounded as
\begin{subequations}\label{DNupper}
\begin{align}
   &D_N(f,R_s) \leq \bar D_N(f,R_s),\\  \nonumber
   &\triangleq\frac{N}{f(N+1)^2}+\frac{\log e\lambda^*(f,R_s,N)f}{N}\\ \label{DNupper_1}
   &+\frac{N-1}{2(N+1)}\sqrt{\left(\frac{N}{f(N+1)}\right)^2+4\log e\lambda^*(f,R_s,N)}, 
\end{align}
\end{subequations}
where $\lambda^*(f,R_s,N)\geq 0$ is the unique solution to \eqref{lambda_exist} with $D^N$ in \eqref{dfrs_cons_a} replaced by \eqref{28ab} and $T_i$, $i=0,\dots,N$ in \eqref{dfrs_cons_a} equal to
\begin{equation}\label{uniform_T}
    T_0=T_1=\dots=T_N = \frac{N}{f(N+1)}.
\end{equation}
\end{lemma}
\begin{proof}
Appendix~\ref{proofLemm3}.
\end{proof}

\begin{lemma}\label{Lemm4}
\begin{equation}\label{underlineDET}
    D_{\mathrm{DET}}(f,R_s) = \frac{1}{2f}+\frac{1}{f(2^{2R_s}-1)},
\end{equation}
where \eqref{underlineDET} can be achieved by a uniform sampling policy with sampling intervals equal to
\begin{equation}
    T_i = \frac{1}{f},~ i = 0,1,\dots
\end{equation}
\end{lemma}
\begin{proof}
Appendix~\ref{proofLemm4}.
\end{proof}

\begin{lemma}\label{Lemm5}
\begin{equation}\label{46_l5}
   D_{\mathrm{DET}}(R) = \min_{\substack{f>0,R_s\geq 1\colon\\fR_s\leq R}}D_{\mathrm{DET}}(f,R_s).
\end{equation}
\end{lemma}
\begin{proof}
Appendix~\ref{proofLemm5}.
\end{proof}

Using Lemma~\ref{Lemm5}, we conclude that \eqref{25a_f} in Theorem~\ref{thm6} holds. It remains to minimize $D_{\mathrm{DET}}(f,R_s)$ in \eqref{25a_f} over feasible $f$ and $R_s$ to prove \eqref{25a_fb}.
\begin{subequations}\label{377}
\begin{align} \label{37a}
    D_{\mathrm{DET}}(R) &= \min_{R_s\geq 1} D_{\mathrm{DET}}\left(\frac{R}{R_s},R_s\right)\\ \label{37b}
    & = D_{\mathrm{DET}}(R,1) \\ \label{37c}
    & = \frac{1}{2R}+\frac{1}{3R} = \frac{5}{6R},
\end{align}
\end{subequations}
where \eqref{37a} holds because $D_{\mathrm{DET}}(f,R_s)$ in \eqref{underlineDET} decreases monotonically in $f$ for any given $R_s\geq 1$, and \eqref{37b} holds because $D_{\mathrm{DET}}\left(\frac{R}{R_s},R_s\right)$ increases monotonically as $R_s$ increases in the range $R_s\geq 1$.
Thus, $D_{\mathrm{DET}}(R)$ is achieved at $f=R$, $R_s = 1$. Note that $\frac{1}{2R}$ in \eqref{37c} comes from the sampling distortion and $\frac{1}{3R}$ comes from the causal IDRF for the discrete-time samples.

\subsection{Proof of Theorem~\ref{Thm5}}\label{IV.D}
From \eqref{377}, we conclude that \eqref{Det} holds. Using Lemma~\ref{Lemm4} and \eqref{25a_fb}, we conclude that the uniform sampling policy with sampling frequency $R$ achieves $D_{\mathrm{DET}}(R)$.

\section{Rate-constrained sampling with delays}\label{V}
In our communication scenario in Section \ref{IA}, the codewords are delivered from the encoder to the decoder without delay, and the distortion constraint \eqref{MSE} penalizes any delay at the encoder or the decoder. While those are realistic assumptions in some scenarios of remote tracking and control, in this section we consider how the achievable distortion-rate tradeoffs are affected if those assumptions are weakened. 
\subsection{Delay at the encoder and the decoder}\label{VA}
In the scenario of encoding the entire process for the purpose of preserving it for future, a large delay is permissible. In the extreme, the encoder may wait until the whole input process $\{W_t\}_{t=0}^T$ is observed before coding, and the decoder is allowed to wait until $T$ before estimating the process. This corresponds to the classical scenario of non-causal (block) compression.  The IDRF for this scenario is given by 
\begin{equation}
\begin{aligned}
&D_{\mathrm{noncausal}}(R)=\\
& \lim_{T\rightarrow\infty}\inf_{\substack{P_{\{\hat W_t\}_{t=0}^T|\{ W_t\}_{t=0}^T}\colon\\ \frac{1}{T}I(\{W_t\}_{t=0}^T;\{\hat W_t\}_{t=0}^T)\leq R}} \mathbb E\left[\frac{1}{T} \int_{0}^T(W_t-\hat W_t)^2 dt\right].
 \end{aligned}
\end{equation}
Berger \cite{26} derived the distortion-rate function for the Wiener process using reverse water-filling over the power spectrum of the process,
\begin{equation} \label{Non}
    D_{\mathrm{noncausal}}(R) = \frac{2\log_2 e}{\pi^2R}\quad \text{bits/s}.
\end{equation}
The ODRF continues to be lower-bounded by the IDRF in this non-causal scenario, $D^{\mathrm{op}}_{\mathrm{noncausal}}~(R)\geq D_{\mathrm{noncausal}}~(R)$ (cf. \eqref{oplarger}). As for the achievability, Berger showed that \eqref{Non} can be achieved in the following sense: given a rate $R\geq 0$, and $\epsilon>0$, there exists a code with rate $R+\epsilon$ that achieves the distortion $D_{\mathrm{noncausal}}(R)+\epsilon$. Berger's coding scheme operates as follows \cite{26}: the Wiener process is divided into successive time intervals of a large enough length $T$ seconds. For each interval, the Karhunen-Lo\`{e}ve (KL) coefficients of the process are calculated, and at most $2^{T(R+\epsilon)}$ codewords are used to jointly encode these coefficients with a resulting MSE per second equal to $D_{\mathrm{noncausal}}(R)+\epsilon$. In parallel with the KL expansion coefficients encoding scheme, an integrating delta modulator is employed to encode each endpoint of the length-$T$ intervals with MSE per second $\epsilon$ using $\epsilon$ bits per second.

Comparing $D_{\mathrm{noncausal}}(R)$ in \eqref{Non} with $D^{\mathrm{op}}(R)$ in \eqref{58}, we see that, surprisingly, the optimal zero-delay policy outperforms the best infinite delay one:
\begin{equation}
    \frac{D^{\mathrm{op}}(R)}{D_{\mathrm{noncausal}}(R)} \approx 0.57.
\end{equation}
This is because in zero-delay causal coding, the timing information is free. Indeed, the decoder knows the codeword-generating time stamps that are stopping times of the filtration generated by the Wiener process. In classical noncausal (block) lossy compression, no encoder and decoder synchronization is assumed, and thus the encoder is tasked with encoding both the values of the Wiener process and the time stamps corresponding to these values. In many operational scenarios of remote tracking and control, the encoder and decoder are naturally synchronized, providing free timing information. Since Berger's distortion-rate function in \eqref{Non} does not take that into account, it cannot adequately characterize the fundamental information-theoretic limits in those scenarios.

\subsection{Delay at the decoder}\label{V.D.}
In the scenario of causal coding where some small delay is tolerated but the data is not recorded for storage, e.g. speech communication, one can leverage both the free timing information and the coding delay to improve distortion-rate tradeoffs. A \emph{one sample look-ahead decoder} waits for the next codeword $U_{\tau_{i+1}}$ before estimating $W_t$, $\tau_i\leq t<\tau_{i+1}$, introducing a maximum average delay of $\mathbb{E}(\tau_{i+1}-\tau_i)=\frac{1}{R}$ at the decoder. As we are about to see, this one sample look-ahead decoder greatly reduces the MSE compared to the ODRF obtained in \eqref{58} under causal estimation.

With the encoding policy in Proposition~\ref{prop2}, the decoder is permitted to estimate $W_t$ at time $t'$, $t\leq t'\leq T$ using not only the codewords received before time $t$, but also the extra codewords received during the time $[t,t']$. In the extreme, $t'=T$, the decoder can jointly use all the codewords and codeword-generating time stamps in time horizon $[0,T]$ to recover the Wiener process. Using Wolf and Ziv's decomposition of MSE in \cite{9}, the ODRF with decoder delay can be decomposed as 
\begin{equation}\label{43}
\begin{aligned}
&D^{\mathrm{op}}_{\mathrm{dec~delay}}(R) =\limsup_{T\rightarrow\infty} \inf_{\substack{\pi_T\in \Pi_T\\ f_T\in F_T\colon\\\eqref{comm_cons} }}\frac{1}{T}\mathbb E\bigg[\sum_{i=0}^{N} \int_{\tau_i}^{\tau_{i+1}}(W_t-\bar W_t)^2\\
&+\left(\bar W_t - \hat W_t\right)^2 dt\biggr],  
\end{aligned}
\end{equation}
where $\bar W_t$ is the MMSE estimator of the process at the encoder using the samples and the times that they were taken: for $t\in [\tau_i,\tau_{i+1})$,
\begin{equation}\label{W_t_bar}
    \bar W_t\triangleq\mathbb E[W_t|\{W_{\tau_j}\}_{j=1}^N,\tau^N] = \mathbb E[W_t|W_{\tau_i},W_{\tau_{i+1}},\tau_i,\tau_{i+1}],
\end{equation}
where \eqref{W_t_bar} holds because $W_t - (W_{\tau_i},W_{\tau_{i+1}}\tau_i,\tau_{i+1})-(\{W_{\tau_j}\}_{j=1}^{i-1},\{W_{\tau_j}\}_{j=i+1}^{N},\{\tau_{j}\}_{j=1}^{i-1},\{\tau_{j}\}_{j=i+1}^{N})$ form a Markov chain in that order. Therefore, given all the noiseless samples, $\bar W_t$ only depends on the previous sample and the next sample. In particular, when the samples are taken under a deterministic sampling policy, $(W_{\tau_i},W_t,W_{\tau_{i+1}})$ is a Gaussian random vector, thus $\bar W_t$ in \eqref{W_t_bar} is the linear interpolation between $W_{\tau_i}$ and $W_{\tau_{i+1}}$. $\hat W_t$ is the MMSE estimator of the process at the decoder using all the received information,
\begin{subequations}\label{one_sp}
\begin{align}
    \hat W_t =& \mathbb E[W_t|U^N,\tau^N]\\
             =& \mathbb E[\mathbb E[W_t|\{W_{\tau_j}\}_{j=1}^N,U^N,\tau^N]|U^N,\tau^N]\\ \label{one_sp_c}
             =& \mathbb E[\bar W_t|U^N,\tau^N],
\end{align}
\end{subequations}
where \eqref{one_sp_c} holds due to the Markov chain $W_t-(\{W_{\tau_j}\}_{j=1}^N,\tau^N)-U^N$ and \eqref{W_t_bar}.
Since the one sample look-ahead decoder only waits until the next codeword $U_{i+1}$ is received at $\tau_{i+1}$, $\hat W_t$ is specified to $\mathbb E[\bar W_t|U^{i+1},\tau^{i+1}]$ for $t\in[\tau_i,\tau_{i+1})$.

We append the one sample look-ahead decoder to the optimal encoding policy in Theorem~\ref{Thm4} and calculate the resulting MSE. Under symmetric threshold sampling policies, the samples are not necessarily Gaussian, and the linear interpolation can be suboptimal. Yet, if in \eqref{43} we substitute for $\bar W_t$ a suboptimal estimate $\frac{W_{\tau_{i+1}}+W_{\tau_i}}{2}$, then the resulting the MSE is equal to $\frac{1}{12R}$, a two-fold improvement over \eqref{58}. We append the one sample look-ahead decoder to the uniform sampling policy in Theorem~\ref{Thm5}, and ignore the potential reduction in quantization distortion brought by the decoder's ability to look ahead by one sample. The resulting sampling distortion is $\frac{1}{T}\mathbb E\biggl[\sum_{i=0}^{N} \int_{\tau_i}^{\tau_{i+1}}(W_t-\bar W_t)^2\biggr] = \frac{1}{6R}$, a $3$-fold improvement over the sampling distortion $\frac{1}{2R}$ \eqref{37c} causally attainable with a uniform sampling policy. Thus, the total MSE is at most $\frac{1}{2R}$, a 1.67-fold improvement over \eqref{Det}.

\subsection{Delay at the channel}\label{channel_section}
Consider the communication scenario in Fig.~\ref{Fig1} with a fixed channel delay between the codeword-generating time stamp and the codeword-delivery time stamp. We show that the optimal coding policy remains the SOI code in Theorem~\ref{Thm4}. Denote the channel delay by $\delta\geq 0$. If the sampling time is $\tau_i$, the delivery time is $\tau_i+\delta$. The encoder and the decoder are clock-synchronized. The decoder knows the delivery time and the fixed delay, thus it knows the sampling time since the channel delay is fixed. The distortion is however measured in real time as in Section~\ref{IA} \eqref{MSE} rather than after a delay as in Sections~\ref{VA}~and~\ref{V.D.}.

The MSE that we aim to minimize under the rate constraint \eqref{comm_cons} is given by
\begin{equation}\label{D_new_channel}
    D_{\mathrm{ch}}(R)=\limsup_{T\rightarrow\infty}\inf_{\substack{\pi_T\in \Pi_T\\ f_T\in F_T\colon\\\eqref{comm_cons} }}\frac{1}{T}\mathbb{E}\left[\sum_{i=0}^{N}\int_{\tau_i+\delta}^{\tau_{i+1}+\delta}(W_t-\hat W_t)^2 dt\right],
\end{equation}
where, similar to \cite{4} and \cite{6}, we use the following MMSE decoding policy,
\begin{align}\label{sub_dec_policy}
    \hat W_t\triangleq\mathbb E[W_t|U^i,\tau^i],~t\in[\tau_i+\delta,\tau_{i+1}+\delta).
\end{align}
Unlike Theorem~\ref{Thm4} where we proved that conditioning on the event $t < \tau_{i+1}$ in the decoding policy \eqref{opt_dec} can be ignored to yield \eqref{sim_opt_dec} without loss of optimality, here we do not delve into the issue of whether ignoring the known event $t < \tau_{i+1} + \delta$ in the conditional expectation \eqref{sub_dec_policy} is optimal.
\begin{proposition}\label{prop3}
In causal coding of the Wiener process with a fixed channel delay and decoding policy \eqref{sub_dec_policy}, the optimal sampling and compressing policy remains the SOI coding scheme in Theorem~\ref{Thm4}, and 
\begin{equation}\label{Dch_57}
    D_{\mathrm{ch}}(R)=\frac{1}{6R}+\delta.
\end{equation}
\end{proposition}
\begin{proof}
Appendix~\ref{G}.
\end{proof}

The optimal sampling policy in the fixed-delay scenario coincides with the optimal sampling policy in the delay-free scenario. This differs from the result of \cite{6}, according to which the optimal causal sampling policy for the Wiener process through a channel with an i.i.d. delay $Y_i$ is a symmetric threshold sampling policy,
\begin{align}\label{6_samp_opt}
    \tau_{i+1}=\inf\{t+\tau_i+Y_i:|W_{t+\tau_i+Y_i}-W_{\tau_i}|\geq \beta\},
\end{align}
where $\beta$ is a threshold that depends on the distribution of $Y_i$ and the sampling frequency constraint.
The setting in \cite{6} is different from our setting in this Section~\ref{channel_section}, since the channel is only allowed to serve one sample at a time in a first-in-first-out (FIFO) principle. Because recent samples must wait in a queue before the previous sample is delivered, the optimal encoder in \cite{6} takes a new sample after the previous sample is delivered, whereas in our setting, the encoder may take a new sample after or before the previous sample is delivered. This results in the  policy in \cite{10} attaining a larger MSE in the constant-delay scenario of Proposition~\ref{prop3} than indicated in \eqref{Dch_57}. We also notice that with the random delay, we cannot simply append an SOI compressor to the optimal symmetric threshold policy \eqref{6_samp_opt} to obtain the optimal rate-constrained code. Indeed, the innovation $W_{\tau_{i+1}+Y_i}-W_{\tau_i}$ may not be a binary random variable for all $i=0,1,\dots$ since waiting for the delivery of the previous sample may cause the thresholds not to be hit with equality at the time $\tau_i$.

\section{Rate-constrained event-triggered control}\label{Control_sec}
The SOI coding scheme proposed in Theorem~\ref{Thm4} can also be applied to the following rate-constrained event-triggered control scenario.
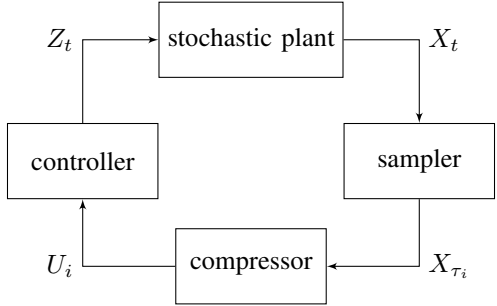
\begin{figure}[h!]
\begin{tikzpicture}[>=latex']
        \tikzset{block/.style= {draw, rectangle, align=center,minimum width=2cm,minimum height=1cm}
        }
        \node [block]  (plant) {stochastic plant};
        \node [block, below right =0.6cm and 0cm of plant] (sampler) {sampler};
        \node [block, below =2cm of plant] (compre) {compressor};
        \node [block, below left =0.6cm and 0cm of plant] (cont) {controller};
        \path[draw,->](plant) -| (sampler) node[midway,right]{$X_t$};
    \path[draw,->](sampler) |- (compre) node[midway,right]{$X_{\tau_i}$};
    \path[draw,->](compre) -|(cont)node[midway,left]{$U_i$};
    \path[draw,->](cont) |-(plant)node[midway,left]{$Z_t$};
    \end{tikzpicture}
    \centering
    \caption{Control system.}
    \label{Fig_cs}
\end{figure}

The stochastic plant evolves according to
\begin{equation}\label{plant}
    dX_t=Z_tdt+dW_t,
\end{equation}
where $W_t$ is the standard Wiener process and $Z_t$ is the control signal generated by the controller. A sampler samples the plant $X_t$ at a sequence of non-decreasing stopping times $\tau_1,\tau_2,\dots$ adapted to the filtration generated by $\{X_t\}_{t=0}^T$. Note that the sampler does not need to know the control signal. At time $\tau_i$, the sampler outputs $X_{\tau_i}$, and the compressor generates codewords $U_i$ based on causally received samples. At time $\tau_i$, the controller uses received codewords $U^i$ to form an impulse control signal $Z_{\tau_i}$. The communication between the compressor and the controller is subject to bits per sec constraint \eqref{comm_cons}. We aim to find the optimal sampling and compressing policies such that the mean-square cost of $X_t$ from target state $0$ is minimized
\begin{equation}\label{plant_MSE}
    \limsup_{T\rightarrow\infty}\frac{1}{T}\mathbb E\left[\int_{0}^TX_t^2 dt\right].
\end{equation}
We restrict our control signal to be the impulse control as in  \cite{1},\cite{Ben}. The impulse control only takes action at the stopping times \eqref{timestamp} decided by the sampling policy, i.e. $Z_t\neq 0$ if and only if $t=\tau_1,\tau_2,\dots$. The impulse control leads to
\begin{equation}\label{imp_ctl}
    X_{\tau_i^+}=X_{\tau_i}+Z_{\tau_i},~i=0,1,\dots
\end{equation}
where $\tau_i^+$ represents the time just after $\tau_i$ \cite{Soh}.
From \eqref{plant} and \eqref{imp_ctl}, we conclude that
\begin{align}\label{X_t_ctl}
    X_t=X_{\tau_i^+}+W_t-W_{\tau_i},~t\in(\tau_i,\tau_{i+1}].
\end{align}
\begin{theorem}\label{thm_contr}
In the rate-constrained event-triggered control system, the jointly optimal sampling and compressing policy that minimizes  \eqref{plant_MSE} is the SOI coding scheme in Theorem~\ref{Thm4}, and the optimal impulse control signal is
\begin{align}\label{opt_ctl}
    Z_{\tau_i}=-(W_{\tau_i}-W_{\tau_{i-1}}), i=0,1,\dots
\end{align}
The minimum mean-square cost \eqref{plant_MSE} is equal to $\frac{1}{6R}$.
\end{theorem}
\begin{proof}
We calculate a lower bound to the MSE in \eqref{plant_MSE} and show that this lower bound can be achieved by the SOI coding scheme. The MSE in \eqref{plant_MSE} is equal to
\begin{subequations}
\begin{align}\label{plant_MSE_a}
    &\limsup_{T\rightarrow\infty}\frac{1}{T}\mathbb E\left[\sum_{i=0}^{N}\int_{\tau_i}^{\tau_{i+1}}(X_{\tau_i^+}+W_t-W_{\tau_i})^2dt\right]\\ \label{plant_MSE_b}
    =&\limsup_{T\rightarrow\infty}\frac{1}{T}\mathbb E\left[\sum_{i=0}^{N}\int_{\tau_i}^{\tau_{i+1}}X_{\tau_i^+}^2dt\right]\\ \label{plant_MSE_c}
    +&\limsup_{T\rightarrow\infty}\frac{1}{T}\mathbb E\left[\sum_{i=0}^{N}\int_{\tau_i}^{\tau_{i+1}}(W_t-W_{\tau_i})^2dt\right]\\ \label{plant_MSE_d}
    +&2\limsup_{T\rightarrow\infty}\mathbb E\left[\sum_{i=0}^NX_{\tau_i^+}\int_{\tau_i}^{\tau_{i+1}}(W_t-W_{\tau_i})dt\right]\\ \label{plant_MSE_e}
    \geq & \limsup_{T\rightarrow\infty}\frac{1}{T}\mathbb E\left[\sum_{i=0}^{N}\int_{\tau_i}^{\tau_{i+1}}(W_t-W_{\tau_i})^2dt\right],
\end{align}
\end{subequations}
where \eqref{plant_MSE_a} is obtained by substituting \eqref{X_t_ctl} into \eqref{plant_MSE}; \eqref{plant_MSE_d} is equal to zero since $X_{\tau_i^+}$ is independent of $\int_{\tau_i}^{\tau_{i+1}}(W_t-W_{\tau_i})dt$ and $\mathbb E\left[\int_{\tau_i}^{\tau_{i+1}}(W_t-W_{\tau_i})dt\right]=0$ for all $i=0,1,\dots$ by the reflection principle of the Wiener process \cite[Chap.3, Thm. A44]{MC}. To achieve the lower bound in \eqref{plant_MSE_e}, we need $X_{\tau_i^+}=0$ for all $i=0,1,\dots$, thus, the optimal impulse control signal is
\begin{equation}
    Z_{\tau_i}=-X_{\tau_i},~i=0,1,\dots,
\end{equation}
which is equal to \eqref{opt_ctl}. Using the arguments in the paragraph below \eqref{NL} in the proof of Theorem~\ref{Thm4}, we can easily verify that the SOI coding scheme achieves the lower bound and is therefore a jointly optimal sampling and compressing policy.
\end{proof}

In contrast, if the sampler samples $\{X_t\}_{t=0}^T$ with a uniform sampling policy that satisfies the rate constraint \eqref{comm_cons} together with some succeeded compressing policy, then we show that the minimum achievable mean-square cost, denoted by $d_{\mathrm{uniform}}$, is lower bounded by
\begin{equation}\label{d_uni_l}
   d_{\mathrm{uniform}}> \frac{1}{2R}.
\end{equation}

To show \eqref{d_uni_l}, we first notice that \eqref{plant_MSE_a}--\eqref{plant_MSE_e} holds under a uniform sampling policy, with the stopping times $\tau_1,\tau_2,\dots$ and random number of samples $N$ replaced by deterministic times $t_1,t_2,\dots$, and the fixed number of samples $n$. We minimize the lower bound \eqref{plant_MSE_e} over all uniform sampling policies that satisfy the following constraint
\begin{equation}
    \frac{n}{T}\leq R,
\end{equation}
and obtain that the minimum of \eqref{plant_MSE_e} is equal to $\frac{1}{2R}$. This lower bound corresponds to the scenario that we ignore the quantization effect due to the compressing policy. The lower bound is not achievable since the controller cannot output the ideally optimal impulse control signal \eqref{opt_ctl} to make $X_{t_i^+}$ zero for all $i=0,1,\dots$ This is because the sample innovation $W_{t_i}-W_{t_{i-1}}$, $i=1,2,\dots$ is a Gaussian random variable that cannot be noiselessly compressed using $1$ bit.

\section{Conclusion}\label{conclusion}
The results in this paper contribute to the rich literature
on optimal scheduling and causal sequential estimation problems
by introducing a transmission rate constraint beyond
the popular sampling frequency constraint. The SOI coding
scheme is optimal for causal estimation of the Wiener
process under an expected rate constraint (Theorem~\ref{Thm4}).
The performance of the SOI coding scheme is much better
than that of the best non-causal code (Section~\ref{VA}). This
underscores the power of free information contained in the
codeword arrival times that is not considered in the standard
setting of non-causal (block) compression. The SOI scheme
 remains optimal even if the channel introduces a fixed delay (Proposition \ref{prop3}).  The key to transmit information via timing is to use process-dependent,
rather than deterministic, sampling time stamps, because the
latter contain zero information. The optimal deterministic
sampling policy is uniform (Theorem~\ref{Thm5}). In either setting,
the best strategy is to transmit lowest possible rate (1-bit
codewords) as frequently as possible (Theorem~\ref{thm6}). This is
a consequence of the real-time distortion constraint \eqref{MSE}. If
a delay is affordable, the MSE can be further reduced with
only one sample look-ahead at the decoder (Section~\ref{V.D.}). The  SOI coding scheme also minimizes the mean-square cost of a stochastic plant driven by the Wiener process, and controlled via impulse control (Theorem~\ref{thm_contr}).

\appendix

\subsection{Proof of Proposition~\ref{prop2}}\label{A}
The objective function in \eqref{1} decomposes in the following way.
\begin{subequations}
\begin{align} \label{52b}
    &\frac{1}{T}  \mathbb E \biggl[\sum_{i=0}^{N} \int_{\tau_i}^{\tau_{i+1}} (W_t - \hat{W}_{\tau_i})^2 dt\biggr]\\ \label{52c}
    =&\frac{1}{T} \mathbb E \biggl[\sum_{i=0}^{N} \int_{\tau_i}^{\tau_{i+1}} (W_t - W_{\tau_i})^2\biggr]+\\ \nonumber &\frac{1}{T} \mathbb E \biggl[\sum_{i=0}^{N}(\tau_{i+1}-\tau_i)(W_{\tau_i}-\hat{W}_{\tau_i})^2 dt\biggr]+\\ \nonumber
    &\frac{1}{T} \mathbb E \biggl[\sum_{i=0}^{N}(W_{\tau_i}-\hat{W}_{\tau_i}) \int_{\tau_i}^{\tau_{i+1}} (W_t - W_{\tau_i}) dt\biggr]\\ \label{10c}
    =&\frac{1}{T} \mathbb E \bigg[\sum_{i=0}^{N} \int_{\tau_i}^{\tau_{i+1}} (W_t - W_{\tau_i})^2\biggr)+\\ \nonumber &\frac{1}{T} \mathbb E \biggl[\sum_{i=0}^{N}(\tau_{i+1}-\tau_i)(W_{\tau_i}-\hat{W}_{\tau_i})^2 dt\biggr],
\end{align}
\end{subequations}
where \eqref{52b} uses the simplified decoding policy \eqref{sim_opt_dec} that is justified in the proof of Theorem~\ref{Thm4}; \eqref{52c} is obtained by substituting $W_t - W_{\tau_i}+W_{\tau_i}- \hat{W}_{\tau_i}$ for the term $W_t - \hat{W}_{\tau_i}$ in \eqref{52b}, and \eqref{10c} holds due to the fact that $\int_{\tau_i}^{\tau_{i+1}} (W_t-W_{\tau_i}) dt$ is orthogonal to $W_{\tau_i}- \hat{W}_{\tau_i}$ for all $i=0,1,2,\dots,N$. Since the encoder only influences the second term in \eqref{10c}, we move the minimization over the encoder $f_T$ in \eqref{1} directly in front of the second term in \eqref{10c}.

To show that $\mathsf f _i$ only encodes $W_{\tau_i}-\hat W_{\tau_{i-1}}$ given $U^{i-1}$ and $\tau^i$, we first recall a well-known fact. Consider the following lossy source coding model in Fig.~\ref{X-Y-Z}, where $X\in\mathcal{X}$ and $Y\in\mathcal{Y}$ are available only at the encoder, $C$ is the common information, $\hat{X}\in \hat{\mathcal X}$ is the reproduction. Encoder $P_{U| X, Y, C}$ and decoder $P_{\hat X| U, C}$ aim to achieve a given distortion $d=\mathbb E\left[\mathsf d(X, \hat X)\right]$, where $\mathsf {d} \colon \mathcal X \times \hat {\mathcal X} \rightarrow \mathbb R^+$ is the distortion measure, subject to a constraint on the cardinality of the alphabet $\mathcal U$ of $U$.  Since  
\begin{equation}
\begin{aligned}
&\mathbb E \left[\mathsf d(X, \hat X) | C = c\right] =\int_{x \in \mathcal X} dP_{X|C = c}(x)\cdot\\ &\int_{U \in \mathcal U} dP_{U | X, C = c}(u)\int_{\hat x \in \hat{\mathcal X}} dP_{\hat X | U, C = c} (\hat x) \mathsf d(x, \hat x),
\end{aligned}
\end{equation}
the knowledge of side information $Y$ is useless at the encoder, i.e. for any encoder-decoder pair $\left(P_{U|X, Y, C},~ P_{\hat X | U, C}\right)$, the pair $\left(P_{U|X, C}, P_{\hat X | U, C}\right)$, where $P_{U|X, C}$ is the marginal of $P_{U|X, Y, C} P_{Y | X, C}$, achieves the same expected distortion. 

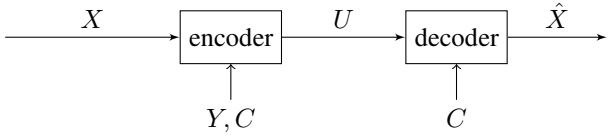
\begin{figure}[h!]
\centering
\tikzstyle{int}=[draw, fill=white!20, minimum size=2em]
\tikzstyle{init} = [pin edge={to-,thin,black}]
\begin{tikzpicture}[node distance=3cm,auto,>=latex']
    \node [int, pin={[init]below:$Y,C$}] (a) {encoder};
    \node (b) [left of=a,node distance=3cm, coordinate] {a};
    \node [int,pin={[init]below:$C$}] (c) [right of=a] {decoder};
    \node [coordinate] (end) [right of=c, node distance=2cm]{};
    \path[->] (b) edge node {$X$} (a);
    \path[->] (a) edge node {$U$} (c);
    \draw[->] (c) edge node {$\hat X$} (end) ;
\end{tikzpicture}
\caption{$Y$ only available at the encoder, $C$ available at both the encoder and the deocder }
\label{X-Y-Z}
\end{figure}
 
In our problem, the infimum of the long-term average MSE in \eqref{10c} corresponds to a causal sampling policy and a sequence of sample distortion allocations $\mathbb E[W_{\tau_{i}}-\hat W_{\tau_{i}}]^2$, $i=1,2,\dots$.  At time $\tau_1$, we take $X = W_{\tau_1}$, $Y = \{W_t\}_{0<t<\tau_1}$, $C = \tau_1$, and $d(X,\hat X)=(X-\hat X)^2$. To achieve a given sample distortion $\mathbb E\left[(W_{\tau_1}-\hat W_{\tau_1})^2\right]$, the random compressing policy needs to only take into account $W_{\tau_1}$ and $\tau_1$. Inductively, at time $\tau_{i}$, the encoder knows $\{W_t\}_{t=0}^{\tau_{i}}$. Both the encoder and the decoder know $U^{i-1}$ and $\tau^{i}$. Since $\hat W_{\tau_{i-1}}$ is known once $U^{i-1}$ and $\tau^{i-1}$ are given, 
\begin{equation}
\begin{aligned}
&\mathbb E[W_{\tau_{i}}-\hat W_{\tau_{i}}]^2=\\
&\mathbb E\left[\left(W_{\tau_{i}}-\hat W_{\tau_{i-1}}-\mathbb E\left[W_{\tau_{i}}-\hat W_{\tau_{i-1}}|U^i,\tau^i\right]\right)^2\right]. \label{48b}
\end{aligned}
\end{equation}
Take $X=W_{\tau_{i}}-\hat W_{\tau_{i-1}}$, $\hat X = \mathbb E[W_{\tau_{i}}-\hat W_{\tau_{i-1}}|U^i,\tau^i]$, $U= U_{i}$, $C= \{U^{i-1},\tau^i\}$, and $Y$ is everything known at the encoder excluding $X$ and $C$. It follows that for the purpose of achieving the sample distortion $\mathbb E[W_{\tau_{i}}-\hat W_{\tau_{i}}]^2$, at time $\tau_i$, the randomized compressing policy needs to only take into account $W_{\tau_{i}}-\hat W_{\tau_{i-1}}$, $U^{i-1}$, and $\tau^i$.

\subsection{Proof of Lemma~\ref{6}}\label{D}
Denote the IDRF for discrete-time samples of the Wiener process
\begin{equation}\label{GM}
    W_{\tau_{i+1}} =W_{\tau_i} + V_{\tau_i},~V_{\tau_i}\sim \mathcal N(0,T_i)
\end{equation}
by $\tilde D_N(R_s)$ \eqref{idfrf_b}. Using the representation of its dual in \cite[Eq. (18)]{10} derived using  a semi-definite programming approach, we represent $\tilde D_N(R_s)$ as
\begin{equation} \label{DRS_GM}
\tilde D_N(R_s)=\inf_{\substack{D_i\geq0,~i=1,\dots,N\colon\\D_{i-1}+T_{i-1}\geq D_i,\quad i=1,2,\cdots,N,\\ 
          \\ \frac{1}{N}\left(\sum_{i=1}^{N}  \frac{1}{2}\log(D_{i-1}+T_{i-1})-\frac{1}{2}\log D_i\right)\leq R_s.}}  \sum_{i=1}^{N}T_iD_i.
\end{equation}

Since the sampling intervals $T^N$ are deterministic, we calculate the summand in \eqref{idfrf_a} as
\begin{equation} \label{trans_T_i}
\mathbb E\biggl[\int_{\tau_{i}}^{\tau_{i+1}}(W_t-W_{\tau_i})^2dt\biggr]
=\mathbb E\biggl[\int_{0}^{T_i} W_t^2dt\biggr] =\frac {T_i^2}{2}.
\end{equation}
Plugging \eqref{DRS_GM} and \eqref{trans_T_i} into \eqref{idfrf}, we can write $D_{\mathrm{DET}}(f,R_s)$ as 
\begin{equation}\label{81_D}
    D_{\mathrm{DET}}(f,R_s) = \limsup_{T\rightarrow\infty}\inf_{\substack{\pi_T\in \Pi^{\mathrm{DET}}_{T}\\ \eqref{f}}} \frac{1}{T}\left(\sum_{i=0}^N\frac{T_i^2}{2}+\tilde D_N(R_s)\right).
\end{equation}
Note that when $T\rightarrow\infty$, the number of samples $N$ must increase no slower than $\sqrt{T}$. Indeed, since the largest sampling interval satisfies 
\begin{equation}
    \max_{i=0,\dots,N} T_i \geq \frac{T}{N+1},
\end{equation}
the summand in \eqref{81_D}
\begin{equation}
    \frac{\max_i T_i^2}{2T} \geq \frac{T}{2(N+1)^2}
\end{equation}
will blow up to infinity if $N$ increases slower than $\sqrt{T}$. Thus, $N\rightarrow\infty$ as $T\rightarrow\infty$. Therefore, we can replace the $\limsup_{T\rightarrow\infty}$ in \eqref{81_D} by $\limsup_{N\rightarrow\infty}$ and obtain \eqref{24}, where we replace $T$ in \eqref{81_D} by $\frac{f}{N}$ as permitted by \eqref{f}, and we replace the minimization constraint \eqref{f} in \eqref{81_D} by its equivalent \eqref{f_equi}.

\subsection{Proof of Lemma~\ref{Lemm2}}\label{proofLemm2}
We split $D_N(f,R_s)$ \eqref{24bb} into the following optimization problems: 
\begin{subequations}
\begin{align}\label{37bb}
& D_N(f,R_s) \triangleq \inf_{\substack{T_0\geq 0,T_N\geq 0 \colon\\ T_0+T_N\leq \frac{N}{f}}} D_N(f,R_s,T_0,T_N),\\  \nonumber
& D_N(f,R_s,T_0,T_N)\triangleq\\ \label{37cc}
&\min_{\substack{T_1,\dots,T_{N-1}\geq 0\colon\\ \frac{1}{N}\sum_{i=1}^{N-1} T_i = \frac{1}{f}-\frac{T_0+T_N}{N}}}\frac{f}{N}\left(\sum_{i=0}^{N}\frac{T_i^2}{2}+ D_N\left(f,R_s,T^N\right)\right),\\ \label{25c}
& D_N\left(f,R_s,T^N\right)\triangleq \min_{\substack{D^N\geq 0\colon\\\eqref{dfrs_cons}}} \sum_{i=1}^{N}T_iD_i.
\end{align}
\end{subequations}
Denote by $\underline{D}_N(f,R_s,T^N)$ the lower bound to $D_N(f,R_s,T^N)$ obtained by deleting the minimization constraint \eqref{min2} in \eqref{25c}, i.e.
\begin{equation}\label{DLfRT}
    \underline{D}_N(f,R_s,T^N) \triangleq \min_{\substack{D^N\geq 0\colon\\\eqref{dfrs_cons_a}}} \sum_{i=1}^{N}T_iD_i,
\end{equation}
Denote by $\underline{D}_N(f,R_s,T_0,T_N)$ the corresponding lower bound to $D_N(f,R_s,T_0,T_N)$ in \eqref{37cc}:
\begin{equation}\label{37cc_l}
\begin{aligned}
& \underline{D}_N(f,R_s,T_0,T_N)\triangleq\\
&\min_{\substack{T_1,\dots,T_{N-1}\geq 0\colon\\ \frac{1}{N}\sum_{i=1}^{N-1} T_i = \frac{1}{f}-\frac{T_0+T_N}{N}}}\frac{f}{N}\left(\sum_{i=0}^{N}\frac{T_i^2}{2}+ \underline{D}_N\left(f,R_s,T^N\right)\right).
\end{aligned}
\end{equation}

We will calculate the corresponding lower bound to $D_N(f,R_s)$:
\begin{equation}\label{37bb_l}
 \underline{D}_N(f,R_s) \triangleq \min_{\substack{T_0\geq 0,T_N\geq 0 \colon\\ T_0+T_N\leq \frac{N}{f}}} \underline{D}_N(f,R_s,T_0,T_N).
\end{equation}
We first show that the optimization problem in the right-hand side of \eqref{DLfRT} is a convex optimization problem that satisfies Slater's condition, i.e. strong duality holds. Then, we solve its Lagrangian dual problem to get the optimal $D_1^*\dots,D_N^*$ in \eqref{28ab} that achieve the minimum in the right-hand side of \eqref{DLfRT}, where $\lambda^*(f,R_s,N)\geq 0$ is the unique solution to \eqref{lambda_exist}.

The objective function $\sum_{i=1}^N T_iD_i$ \eqref{DLfRT} is an affine function in $D^N$. Furthermore, $z\left(D^N\right)$ is a convex function since
\begin{subequations}
\begin{align}
    &\frac{\partial^2z\left(D^N\right)}{\partial D_i^2} = \frac{\log e T_i(2D_i+T_i)}{N(D_i^2+D_iT_i)^2}\geq 0,~\forall i=1,\dots,N-1,\\
    &\frac{\partial^2z\left(D^N\right)}{\partial D_N^2} = \frac{\log e}{ND_N^2}\geq 0,\\
    &\frac{\partial^2z\left(D^N\right)}{\partial D_i \partial D_j} = 0,~\forall i,j = 1,\dots,N.
\end{align}
\end{subequations}
Therefore, the minimization problem in the right-hand side of \eqref{DLfRT} is convex. Notice that $z(D,D,\dots,D)$ decreases from $+\infty$ to $-\infty$ as $D$ increases from $0$ to $\infty$. Thus, there exists $\tilde D\geq 0$ such that Slater's condition is satisfied, i.e. 
\begin{equation}\label{slater}
    z\left(\tilde{D},\tilde{D},\dots,\tilde{D}\right)<2R_s.
\end{equation}
We conclude that 1) the strong duality holds, 2) $\underline{D}(f,R_s,T^N)$ can be obtained via its Lagrangian dual problem, and 3) there must exist an optimal Lagrangian multiplier $\lambda^*(f,R_s,N)\geq 0$ that satisfies the complementary slackness \eqref{lambda_exist} in the Karush-Kuhn-Tucker conditions. Indeed, \eqref{lambda_exist} always has a non-negative solution $\lambda^*(f,R_s,N)$, since as a function of $\lambda^*(f,R_s,N)$, $z\left(D^{N*}\right)$ is continuous and monotonically decreasing from $+\infty$ to $-\infty$ as $\lambda^*(f,R_s,N)$ increases from $0$ to $+\infty$.

Plugging $D^{N*}$ \eqref{28ab} into \eqref{DLfRT}, we obtain $\underline{D}_N\left(f,R_s,T^N\right)$ and proceed to evaluate $\underline{D}_N(f,R_s, T_0, T_N)$ in \eqref{37cc_l}, which is given by
\begin{equation}\label{64_follow}
\underline{D}_N(f,R_s,T_0,T_N)=\min_{\substack{T_1,\dots,T_{N-1} \geq 0 \colon\\ \frac{1}{N}\sum_{i=1}^{N-1} T_i = \frac{1}{f}-\frac{T_0+T_N}{N}}}  g(T_1,\dots,T_{N-1}),
\end{equation}
where
\begin{equation}\label{g}
\begin{aligned}
    &g(T_1,\dots,T_{N-1}) \triangleq \frac{f}{2N}\bigg[T_0^2+T_N^2+2\log e\lambda^*(f,R_s,N)\\ &+\sum_{i=1}^{N-1} T_i\sqrt{T_i^2+4\log e\lambda^*(f,R_s,N)}\Biggr).
\end{aligned}
\end{equation}
We make use of the Schur-convexity of \eqref{g} to calculate $\underline{D}_N(f,R_s,T_0,T_N)$. Recall that if a function $f(x^d)$ is symmetric and its first partial derivative with respect to each $x_i$, $i=1,\dots,d$ exits, then $f(x^d)$ is Schur-convex if and only if
\begin{equation}\label{schur_test}
   (x_i-x_j)\left(\frac{\partial f(x^d)}{\partial x_i}-\frac{\partial f(x^d)}{\partial x_j}\right)\geq 0,~ \forall~i,j=1,\dots,d.
\end{equation}

It is clear that $g(T_1,\dots,T_{N-1})$ is symmetric since it is invariant to the permutations of $T_1,\dots,T_{N-1}$. 
To calculate the partial derivatives of \eqref{g}, we first compute the implicit differentiation $\frac{\partial \lambda^*(f,R_s,N)}{\partial T_i}$ by taking the derivative with respect to $T_i$ on the both sides of \eqref{lambda_exist}, yielding
\begin{equation}\label{lambda_deri}
\begin{aligned}
    &\frac{\partial \lambda^*(f,R_s,N)}{\partial T_i} =\frac{1}{\sqrt{T_i^2+4\log e\lambda^*(f,R_s,N)}}\cdot\\ &\frac{2\lambda^*(f,R_s,N)}{1+\sum_{k=1}^{N-1}\frac{T_k}{\sqrt{T_k^2+4\log e\lambda^*(f,R_s,N)}}}.
\end{aligned}
\end{equation}
Using \eqref{lambda_deri} to compute the first partial derivative, we obtain 
\begin{subequations}\label{fpd}
\begin{align}
   &\frac{\partial g(T_1,\dots,T_{N-1})}{\partial T_i}\\\nonumber =&\frac{f}{2N}\Biggl(2\log e\frac{\partial \lambda^*(f,R_s,N)}{\partial T_i}+\sqrt{T_i^2+4\log e\lambda^*(f,R_s,N)}\\
   +&\frac{T_i^2+2\log eT_i\frac{\partial \lambda^*(f,R_s,N)}{\partial T_i}}{\sqrt{T_i^2+4\log e\lambda^*(f,R_s,N)}}\\+& \sum_{\substack{k=1\\k\neq i}}^N \frac{2\log eT_k\frac{\partial \lambda^*(f,R_s,N)}{\partial T_i}}{\sqrt{T_k^2+4\log e\lambda^*(f,R_s,N)}}\Biggr)\\
   =&\frac{f}{N}\sqrt{T_i^2+4\log e\lambda^*(f,R_s,N)}.
\end{align}
\end{subequations}
Using \eqref{fpd}, we can verify that $g(T_1,\dots,T_{N-1})$ satisfies \eqref{schur_test}: 
\begin{equation}\label{788}
\begin{aligned}
    &(T_i-T_j)\frac{f}{N}\cdot\Biggl(\sqrt{T_i^2+4\log e\lambda^*(f,R_s,N)}\\&-\sqrt{T_j^2+4\log e\lambda^*(f,R_s,N)}\Biggr)\geq 0, 
\end{aligned}
\end{equation}
for all $i,j=1,\dots,N-1$. Therefore, $g(T_1,\dots,T_{N-1})$ is a Schur-convex function. 

Let $x=(x_1,\dots,x_d)\in \mathbb R^d$, $y = (y_1,\dots,y_d)\in \mathbb R^d$ be two non-increasing sequences of real numbers. Recall that $x$ is majorized by $y$ if for each $k=1,\dots,d$, $\sum_{i=1}^k x_i\leq \sum_{i=1}^k y_i$ with equality if $k=d$. For a Schur-convex function $f$, if $x$ is majorized by $y$, then $f(x)\leq f(y)$. In our case, the feasible $T_i$'s must satisfy the minimization constraint of the optimization problem in \eqref{64_follow}.  Any sequence $T_1,\dots,T_{N-1}$ that satisfies the minimization constraint of the optimization problem in \eqref{64_follow} majorizes the sequence in \eqref{T_seq}.  Therefore, the infimum in \eqref{64_follow} is achieved by the sequence $T_1^*,\dots,T_{N-1}^*$ in \eqref{T_seq}.

Plugging $T_1^*,\dots,T_{N-1}^*$ \eqref{T_seq} into \eqref{64_follow}, we obtain
\begin{equation}\label{underDnfrsTT}
\begin{aligned}
   &\underline{D}_N(f,R_s,T_0,T_N) = \frac{f}{2} \Biggl(\frac{T_0^2+T_N^2+2\log e\lambda^*(f,R_s,N)}{N}\\ &+\frac{N-1}{N}T^*(f,N)\sqrt{T^*(f,N)^2+4\log e\lambda^*(f,R_s,N)}\Biggr). 
\end{aligned}
\end{equation}
Plugging \eqref{underDnfrsTT} into the right-hand side of  \eqref{37bb_l} completes the proof.

\subsection{Proof of Lemma~\ref{Lemm3}}\label{proofLemm3}
Plugging \eqref{uniform_T} into \eqref{28ab}, we obtain the corresponding optimal sample distortions,
\begin{subequations}\label{D_f}
\begin{align}\nonumber
    &D^*_1=\dots=D^*_{N-1} =\\\label{D_f_a} &\frac{-\frac{N}{f(N+1)}+\sqrt{\left(\frac{N}{f(N+1)}\right)^2+4\log e\lambda^*(f,R_s,N)}}{2},\\
    &D^*_N  = \frac{f(N+1)}{N}\log e\lambda^*(f,R_s,N),
\end{align}
\end{subequations}
where $\lambda^*(f,R_s,N)$ is defined in Lemma~\ref{Lemm3}.

We first show that the $T^N$ in \eqref{uniform_T} and the corresponding $D^N$ in \eqref{D_f} satisfy the deleted constraint \eqref{min2}, then we can plug $T^N$ \eqref{uniform_T} and $D^N$ \eqref{D_f} as feasible solutions into the minimization problem associated with $D_N(f,R_s)$ in \eqref{24bb} to obtain the upper bound in \eqref{DNupper}.

When $i=2,\dots,N-1$, the deleted constraint \eqref{min2} is satisfied trivially, since $D_{i-1}=D_{i}$ and $T_{i-1}\geq 0$. To prove that the deleted constraint \eqref{min2} also holds at $i=1$ and $N$, we upper bound $\lambda^*(f,R_s,N)$ for every $N > 2$. When 
\begin{equation}
    T_1=\dots=T_{N-1},
\end{equation}
we can rearrange terms in the complementary slackness condition \eqref{lambda_exist} and conclude $x = \lambda^*(f,R_s,N)\log e$ is the unique solution to the following equation,
\begin{equation}\label{HN_X}
    h_N(T_0,T_N,T_1,R_s,x) -x = 0,
\end{equation}
where
\begin{equation}\label{H_N}
\begin{aligned}
&h_N(T_0,T_N,T_1,R_s,x) \triangleq\\
&\frac{T_1^2}{2^{2R_{s}+\frac{2}{N-1}R_{s}-\frac{\log T_0+\log T_N}{N-1}+\frac{\log x}{N-1}}-1}\\
+& \left(\frac{T_1}{2^{2R_{s}+\frac{2}{N-1}R_{s}-\frac{\log T_0+\log T_N}{N-1}+\frac{\log x}{N-1}}-1}\right)^2.
\end{aligned}
\end{equation}
Note that the left-hand side of \eqref{HN_X} monotonically decreases as $x$ increases.

Given $R_s$, plugging \eqref{uniform_T} into the left-hand side of \eqref{HN_X}, we conclude that the $\lambda^*(f,R_s,N)$ in Lemma~\ref{Lemm3} is the unique solution to the following equation, 
\begin{equation}\label{equa1}
    h_N\left(\frac{N}{f(N+1)},\frac{N}{f(N+1)},\frac{N}{f(N+1)},R_s,x\right)-x=0,
\end{equation}
 Plugging
\begin{equation}
    x = \frac{N^2}{2f^2(N+1)^2}
\end{equation}
into \eqref{equa1}, we observe that the left-hand side of \eqref{equa1} is less or equal to $0$ for all $N>2$. Thus, we conclude
\begin{equation}\label{ub_l}
   \lambda^*(f,R_s,N)\log e\leq \frac{N^2}{2f^2(N+1)^2},~\forall~N>2.
\end{equation}
Plugging \eqref{ub_l} into \eqref{D_f}, we obtain
\begin{subequations}\label{D1N}
\begin{align}
&D_1^* \leq \sqrt{\lambda^*(f,R_s,N)\log e}\leq \frac{N}{f(N+1)}, \\
&D_N^* \leq \frac{N}{2f(N+1)},
\end{align}
\end{subequations}
Substituting \eqref{uniform_T} and \eqref{D1N} into \eqref{min2}, we conclude that \eqref{min2} holds for $i=1$ and $i=N$. 

Now, we can plug \eqref{uniform_T} and \eqref{D_f} as feasible solutions into \eqref{24bb} to obtain the right-hand side of \eqref{DNupper}.

\subsection{Proof of Lemma~\ref{Lemm4}}\label{proofLemm4}
From Lemmas~\ref{Lemm2}~and~\ref{Lemm3}, and \eqref{37aa},
\begin{equation}\label{new_display}
    \liminf_{N\rightarrow\infty}\underline{D}_N(f,R_s)\leq D_{\mathrm{DET}}(f,R_s) \leq \limsup_{N\rightarrow\infty}\bar D_N(f,R_s).
\end{equation}
We prove \eqref{underlineDET} by showing that both bounds are equal to the right-hand side of \eqref{underlineDET}.

To compute the lower bound in \eqref{new_display}, we need to understand the behavior of  $T^*(f,N)$, $\lambda^*(f,R_s,N)$ and $T_0^*$, $T_N^*$ as $N$ goes to infinity, where $T_0^*$, $T_N^*$ achieve the minimum of the left-hand side of \eqref{new_display}. $T_0^*$ and $T_N^*$ must increase as
\begin{equation}\label{on}
  T_0^*+T_N^* = O\left(\sqrt{N}\right),  
\end{equation}
or $\frac{{T_0^*}^2+{T_N^*}^2}{N}$ in \eqref{38bb} will blow up to infinity as $N\rightarrow\infty$. Substituting \eqref{on} to \eqref{T_seq}, we obtain
\begin{equation}\label{T_i_n1}
    T^*(f,N)=\frac{1}{f}+O\left(\frac{1}{\sqrt{N}}\right).
\end{equation}

We proceed to compute
\begin{equation}\label{lam}
    \lambda^* \triangleq \lim_{N\rightarrow\infty}\lambda^*(f,R_s,N).
\end{equation}
For given $T_0^*$, $T_N^*$ and $R_s$, $x=\lambda^*(f,R_s,N)\log e$ is the unique solution to \eqref{HN_X} with $T_0$, $T_N$, and $T(N)$ replaced by $T_0^*$, $T_N^*$ and $T^*(f,N)$ in \eqref{T_seq}.
We prove that
\begin{subequations}\label{lamm}
\begin{align}\label{lamm1}
&\lambda^*\log e\geq \frac{1}{2^{2R_s}f^2}, \\\label{lamm2}
&\lambda^*\log e\leq \frac{1}{2f^2}.
\end{align}
\end{subequations}
We substitute \eqref{on} and \eqref{T_i_n1} into the left-hand side of \eqref{HN_X} and take $\lim_{N\rightarrow\infty}$ to conclude that
\begin{equation} \label{up_lo}
   \lim_{N\rightarrow\infty} h_N\left(T_0^*,T_N^*,T^*(f,N),R_s,\frac{1}{2f^2}\right)-\frac{1}{2f^2}\leq 0.
\end{equation}
Using the fact that the left-hand side of \eqref{HN_X} is monotonically decreasing in $x$, we conclude \eqref{lamm1} holds.
To prove \eqref{lamm2}, we similarly compute
\begin{equation}\label{lo_up}
     \lim_{N\rightarrow\infty}h_N\left(T_0^*,T_N^*,T^*(f,N),R_s,\frac{1}{2^{2R_s}f^2}\right)-\frac{1}{2^{2R_s}f^2}\geq 0.
\end{equation}

Via the squeeze theorem, \eqref{lamm} implies
\begin{equation}\label{log_lamb1}
    \lambda^*(f,R_s,N)= O(1).
\end{equation}
Plugging \eqref{on}, \eqref{T_i_n1} and \eqref{log_lamb1} into \eqref{HN_X}, and taking $N\rightarrow\infty$ on both sides of \eqref{HN_X}, we obtain 
\begin{equation}\label{70_up}
\lambda^*\log e=\frac{1}{f^2(2^{2R_s}-1)^2} + \frac{1}{f^2(2^{2R_s}-1)}.
\end{equation}

Plugging \eqref{on}, \eqref{T_i_n1} and \eqref{70_up} into the right-hand side of \eqref{38bb} and taking $\lim_{N\rightarrow\infty}$, we compute
\begin{subequations} \label{up}
\begin{align}\nonumber
     &\lim_{N\rightarrow\infty}
   \underline{D}_N(f,R_s)\\ \label{sum22}
    &= \frac{1}{2f}+\frac{1}{f(2^{2R_s}-1)}+\lim_{N\rightarrow\infty}\inf_{\substack{T_0\geq0,T_N\geq 0\\ T_0+T_N\leq \frac{N}{f}}}
    \frac{f}{2}\left(\frac{T_0^2+T_N^2}{N}\right)\\ \label{97cc}
    &= \frac{1}{2f}+\frac{1}{f(2^{2R_s}-1)},
\end{align}
\end{subequations}
where $0$ is achieved in the last term of \eqref{sum22} by choosing any pair of $T_0, T_N\geq 0$ that satisfies
\begin{equation}\label{final_cons}
   T_0+T_N=o\left(\sqrt{N}\right).
\end{equation}
We choose $T_0$ and $T_N$ in \eqref{uniform_T} that satisfy \eqref{final_cons}, such that together with $T_1,\dots,T_{N-1}$ in \eqref{uniform_T}, the lower bound of $D_{\mathrm{DET}}(f,R_s)$ in \eqref{new_display} is achieved. 

Now, we compute the upper bound in the right-hand side of \eqref{new_display}. $\lambda^*(f,R_s,N)\log e$ in \eqref{DNupper_1} is the unique solution to \eqref{HN_X}. Note that \eqref{70_up} holds for any $T_0$ and $T_N$ that satisfy \eqref{on}. Since $T_0$ and $T_N$ in \eqref{uniform_T} satisfy \eqref{on}, we conclude that the $\lim_{N\rightarrow\infty}$ of $\lambda^*(f,R_s,N)\log e$ in \eqref{DNupper_1} is also equal to \eqref{70_up}. Plugging \eqref{70_up} into the right-hand side of \eqref{DNupper_1} and taking $\limsup_{N\rightarrow\infty}$, we calculate that the upper bound of $D_{\mathrm{DET}}(f,R_s)$ in \eqref{new_display} is equal to \eqref{97cc}. 

Furthermore, we observe that the uniform sampling intervals \eqref{uniform_T} achieving both the upper and the lower bound of $D_{\mathrm{DET}}(f,R_s)$, converge to $\frac{1}{f}$ asymptotically. We conclude that the uniform sampling policy with the sampling interval $\frac{1}{f}$ achieves $D_{\mathrm{DET}}(f,R_s)$.

\subsection{Proof of Lemma~\ref{Lemm5}}\label{proofLemm5}
The max-min inequality and \eqref{IDRF_IDFRF} imply that
\begin{equation}\label{minlim_upper}
D_{\mathrm{DET}}(R)
\leq \min_{\substack{f>0,R_s\geq 1\colon\\fR_s\leq R}}\limsup_{N\rightarrow\infty}  \bar{D}_N(f,R_s).
\end{equation}
On the other hand,
\begin{subequations}\label{minlim_under}
\begin{align}\label{101a}
D_{\mathrm{DET}}(R) &\geq \lim_{N\rightarrow\infty} \inf_{\substack{f>0,R_s\geq 1\colon\\fR_s\leq R}} \underline{D}_N(f,R_s)\\ \label{101b}
& = \inf_{\substack{f>0,R_s\geq 1\colon\\fR_s\leq R}}\lim_{N\rightarrow\infty}\underline{D}_N(f,R_s),
\end{align}
\end{subequations}
where \eqref{101a} is by \eqref{IDRF_IDFRF}, and \eqref{101b} will be proved in the sequel.
Using \eqref{new_display} with both bounds equal to each other, \eqref{minlim_upper} and \eqref{minlim_under}, we complete the proof of Lemma~\ref{Lemm5}.

We proceed to prove \eqref{101b} via the fundamental theorem of $\Gamma$-convergence. Let $\mathcal X$ be a topological space and $G_N:\mathcal{X}\rightarrow[0,+\infty]$, $N=1,2,\dots,$ be a sequence of functions defined on $\mathcal{X}$. 
A sequence of functions $G_N$, $N=1,2,\dots$ $\Gamma$-converges to its $\Gamma$-limit $G\colon\mathcal X\rightarrow[0,+\infty]$ if \cite{29}:

(i) For every $x\in \mathcal{X}$, and for every sequence $x_N\in\mathcal X, N=1,2,\dots$ converging to $x$,
\begin{equation}
    G(x)\leq \liminf_{N\rightarrow\infty} G_N(x_N).
\end{equation}

(ii) For every $x\in \mathcal{X}$, there exists a sequence $x_N\in\mathcal X, N=1,2,\dots$ converging to $x$ such that
\begin{equation}
    G(x)\geq \limsup_{N\rightarrow\infty} G_N(x_N).
\end{equation}

A sequence of functions $G_N$, $N=1,2,\dots$ is equicoercive \cite{29} if there exists a compact set $\mathcal{K}$ that is independent of $N$, such that 
\begin{equation}\label{equicoercive}
    \inf_{x\in \mathcal{X}} G_N(x) = \inf_{x\in \mathcal{K}}G_N(x).
\end{equation}

The fundamental theorem of $\Gamma$-convergence \cite{29} says that if $G_N$ is equicoercive and $\Gamma$-converges to $G\colon\mathcal X\rightarrow [0,+\infty]$, then we have,
\begin{equation}
    \min_{x\in \mathcal{X}}G(x) = \lim_{N\rightarrow\infty} \inf_{x\in \mathcal{X}}G_N(x).
\end{equation}

We will show that for any scalars $f>0$, $R_s\geq 1$ and for any sequences $f_{(N)}\rightarrow f$, $R_{s(N)}\rightarrow R_s$, we have
\begin{equation}\label{limDN}
    \lim_{N\rightarrow\infty} \underline{D}_N(f_{(N)},R_{s(N)}) = D_{\mathrm{DET}}(f,R_s),
\end{equation}
which means in particular that $D_{\mathrm{DET}}(\cdot,\cdot)$ is the $\Gamma$-limit of $\underline{D}_{N}(\cdot,\cdot)$.
We will also prove that $\underline{D}_N(f,R_s)$ is equicoercive, and \eqref{101b} will follow via the fundamental theorem of $\Gamma$-convergence.

We verify that the reasoning in \eqref{on}-\eqref{up} goes through replacing $f$ and $R_s$ by $f_{(N)}$ and $R_{s(N)}$ respectively, hence \eqref{limDN} holds.

It remains to prove that $\underline{D}_N(f,R_s)$ is equicoercive.
Ignoring the two non-negative $\lambda^*(f,R_s,N)$ terms in the right-hand side of \eqref{38bb}, we observe that
\begin{subequations}
\begin{align}\nonumber
    &\underline{D}_N(f,R_s) \\\label{116b}
    \geq& \inf_{\substack{T_0\geq0,T_N\geq 0\\ T_0+T_N\leq \frac{N}{f}}}
    \frac{f}{2} \left(\frac{T_0^2+T_N^2}{N}+\frac{N-1}{N}T^*(f,N)^2\right)\\ \nonumber 
    =& \inf_{\substack{T_0\geq0,T_N\geq 0\\ T_0+T_N\leq \frac{N}{f}}}\frac{1}{2}\Biggl( f\frac{T_0^2+T_N^2}{N}\\\label{116c}
    +&\frac{N}{f(N-1)}\left(1-\frac{f(T_0+T_N)}{N}\right)^2\Biggr),
\end{align}
\end{subequations}
where \eqref{116c} is obtained by plugging \eqref{T_seq} into \eqref{116b}. Denote the objective function in \eqref{116c} by $q(T_0,T_N)$. We prove that $q(T_0,T_N)$ is a Schur-convex function: 1) $q(T_0,T_N)$ is symmetric, since it is invariant to the permutations of $T_0$ and $T_N$; 2) the first-order partial derivatives of $q(T_0,T_N)$ with respect to $T_0$ and $T_N$ are 
\begin{subequations}\label{q_deriv}
\begin{align}
&\frac{\partial q}{\partial T_0} = \frac{f}{N}T_0+\frac{f}{N(N-1)}(T_0+T_N)-\frac{1}{N-1}, \\
&\frac{\partial q}{\partial T_N} = \frac{f}{N}T_N+\frac{f}{N(N-1)}(T_0+T_N)-\frac{1}{N-1},
\end{align}
\end{subequations}
where \eqref{q_deriv} satisfies \eqref{schur_test}. Using the property of Schur-convex functions stated in Lemma~\ref{Lemm2} after \eqref{788}, we know that the minimum of $q(T_0,T_N)$ is achieved by
\begin{equation}\label{a_resp}
    T_0=T_N = a.
\end{equation}
for some
\begin{equation}\label{2a_cons}
    0\leq a\leq \frac{N}{2f}.
\end{equation}
Plugging \eqref{a_resp} into $q(T_0,T_N)$, and minimizing $q(a,a)$ under the constraint \eqref{2a_cons}, we find that the optimal $a$ that minimizes $q(a,a)$ is given by
\begin{equation}\label{a_equal}
    a = \frac{N}{(N+1)f}.
\end{equation}
Plugging \eqref{a_resp} and \eqref{a_equal} into \eqref{116c}, we obtain
\begin{equation}\label{127}
    \underline{D}_N(f,R_s) \geq \frac{N^2}{2f(N+1)^2}.
\end{equation}
On the other hand, plugging \eqref{ub_l} into the right-hand side of \eqref{476}, we obtain
\begin{equation}\label{128}
\bar D_N(f,R_s) \leq \frac{3N}{2f(N+1)^2}+\frac{\sqrt{3}N(N-1)}{2f(N+1)^2}.
\end{equation}
Choosing $f=R$ in \eqref{128}, we conclude that
\begin{equation}\label{129}
    \inf_{\substack{f>0,R_s\geq 1\\ fR_s\leq R}} \underline{D}_N(f,R_s)\leq \frac{3N}{2R(N+1)^2}+\frac{\sqrt{3}N(N-1)}{2R(N+1)^2}.
\end{equation}
For any
\begin{equation}\label{f_exclude}
    f\in\left(0,\frac{R}{3+\sqrt{3}}\right),
\end{equation}
the right-hand side of \eqref{127} is larger than the right-hand side of \eqref{129}, thus $f$ in \eqref{f_exclude} cannot attain the infimum in \eqref{129}. It follows that the infimum is attaned in the following compact set for $f$, 
\begin{equation}\label{f_compact}
f\in\left[\frac{R}{3+\sqrt{3}}, R\right],
\end{equation}
where the upper bound of $f$ is obtained by lower-bounding $R_s$ by $1$. Correspondingly, $R_s$ lies within the following compact set,
\begin{equation}\label{R_compact}
    R_s\in \left[1,3+\sqrt{3}\right],
\end{equation}
Using \eqref{f_compact} and \eqref{R_compact}, we conclude that $\underline D_N(f,R_s)$ is equicoercive.

\subsection{Proof of Proposition~\ref{prop3}}\label{G}
We derive a lower bound to \eqref{D_new_channel}, and show that the lower bound is achieved by the SOI code.
Note that \eqref{D_new_channel} is lower bounded by
\begin{equation}\label{channel_lb_new_d}
    \limsup_{T\rightarrow\infty}\inf_{\substack{\pi_T\in \Pi_T\colon\\\frac{\mathbb E[N]}{T}\leq R }}\frac{1}{T}\mathbb{E}\left[\sum_{i=0}^{N}\int_{\tau_i+\delta}^{\tau_{i+1}+\delta}(W_t-\tilde W_t)^2 dt\right],
\end{equation}
where 
\begin{align}\label{channel_lb_new}
    \tilde W_t~&\triangleq\mathbb E[W_t|\{W_s\}_{s=0}^{\tau_i},\tau^i]\\\label{channel_lb_new_r}
    &=W_{\tau_i},~t\in[\tau_i+\delta,\tau_{i+1}+\delta),
\end{align}
since $\sigma(U^i)\subset\sigma(\{W_s\}_{s=0}^{\tau_i})$ and \eqref{NL}.
Plugging \eqref{channel_lb_new_r} into the lower bound \eqref{channel_lb_new_d}, we obtain the objective function,
\begin{subequations}
\allowdisplaybreaks
\begin{align}\label{62_real}
&\frac{1}{T}\mathbb{E}\left[\sum_{i=0}^{N}\int_{\tau_i}^{\tau_{i+1}}(W_t-W_{\tau_i})^2 dt\right]\\
-&~\frac{1}{T}\mathbb{E}\left[\sum_{i=0}^{N}\int_{\tau_i}^{\tau_{i}+\delta}(W_t-W_{\tau_i})^2 dt\right]\\
+&~\frac{1}{T}\mathbb{E}\left[\sum_{i=0}^{N}\int_{\tau_{i+1}}^{\tau_{i+1}+\delta}(W_t-W_{\tau_i})^2 dt\right]\\
=&~\frac{1}{T}\mathbb{E}\left[\sum_{i=0}^{N}\int_{\tau_i}^{\tau_{i+1}}(W_t-W_{\tau_i})^2 dt\right]\\
-&~\frac{1}{T}\mathbb E\left[\sum_{i=0}^N\frac{\delta^2}{2}\right]\\
+&~\frac{1}{T}\mathbb E\left[\sum_{i=0}^N\frac{\delta^2}{2}+\delta(\tau_{i+1}-\tau_i)\right]\\\label{62_real_end}
=&~  \frac{1}{T}\mathbb{E}\left[\sum_{i=0}^{N}\int_{\tau_i}^{\tau_{i+1}}(W_t-W_{\tau_i})^2 dt\right]+\delta.
\end{align}
\end{subequations}
Note that the first part of \eqref{62_real_end} is equal to \eqref{thm1_dec_d} in the delay-free case, and $\delta$ is a fixed number. Following the arguments in the paragraph below \eqref{thm1_dec_d}, we conclude that the SOI coding scheme achieves \eqref{channel_lb_new_d}.
\end{document}